\documentclass[a4paper,10.5pt]{article}
\usepackage{amssymb,amsmath,latexsym,amsthm}
\usepackage{color,mathrsfs}
\usepackage{url}
\usepackage{enumerate}
\usepackage[square,sort,comma,numbers]{natbib}
\usepackage{dsfont}
\usepackage{diagbox}
 \usepackage{geometry}
\usepackage{graphicx}
\usepackage{colortbl}
\usepackage{stmaryrd}
\usepackage{endnotes}

\renewcommand{\textbf}[1]{\begingroup\bfseries\mathversion{bold}#1\endgroup}

\usepackage{fancyhdr}
\usepackage{pstricks,pst-plot,pst-node,pstricks-add}

\newlength{\bibitemsep}\newlength{\bibparskip}
\let\oldthebibliography\thebibliography \renewcommand\thebibliography[1]{
  \oldthebibliography{#1} \setlength{\parskip}{\bibitemsep}
  \setlength{\itemsep}{\bibparskip} }

\newtheorem{thm}{Theorem}[section]
\newtheorem{defi}{Definition}[section]
\newtheorem{corollary}[thm]{Corollary}
\newtheorem{prop}[thm]{Proposition}
\newtheorem{Conjecture}[thm]{Conjecture}
\newtheorem{lemma}[thm]{Lemma}

\theoremstyle{definition}
\newtheorem{remark}[thm]{Remark}

\newtheorem{example}[thm]{Example}

\definecolor{DarkGreen}{RGB}{51,153,0}

\makeatletter
\let\@fnsymbol\@arabic
\makeatother

\newcommand{\R}{\mathbb R}

\newcommand{\Z}{\mathbb Z}
\newcommand{\N}{\mathbb N}

\numberwithin{equation}{section}

\def\XXint#1#2#3{{\setbox0=\hbox{$#1{#2#3}{\int}$}
    \vcenter{\hbox{$#2#3$}}\kern-.5\wd0}}

\allowdisplaybreaks
\date{date}
\begin{document}
\title{On crystallization in the plane for pair potentials \\ with an arbitrary norm}
\author{Laurent B\'{e}termin\footnote{Email address: \texttt{betermin@math.univ-lyon1.fr}} \quad Camille Furlanetto\footnote{Email address: \texttt{camille.furlanetto@etu.univ-lyon1.fr}}\\ \\
Institut Camille Jordan $\&$ Universit\'e Claude Bernard Lyon 1\\ 43 boulevard du 11 Novembre 1918, 69622 Villeurbanne Cedex, France }
\date\today
\maketitle

\begin{abstract}
We investigate two-dimensional crystallization phenomena, i.e., minimality of a lattice's patch for interaction energies, with pair potentials of type $(x,y)\mapsto V(\|x-y\|)$ where $\|\cdot\|$ is an arbitrary norm on $\R^2$ and $V:\R_+^*\to \R$ is a function. For the Heitmann-Radin sticky disk potential $V=V_{\textnormal{HR}}$, we prove, using Brass' key result from [\textit{Computational Geometry}, 6:195--214, 1996], that crystallization occurs for any fixed norm, with a classification of minimizers and minimal energies according to the kissing number associated to $\|\cdot\|$. The minimizer is proved to be, up to affine transform, a patch of the triangular or the square lattice, which shows how to easily get anisotropy in a crystallization phenomenon. We apply this result to the $p$-norms $\|\cdot\|_p$, $p\geq 1$, which allows us to construct an explicit family of norms for which crystallization holds on any given lattice. We also solve part of a crystallization problem studied in [\textit{Arch. Ration. Mech. Anal.}, 240:987--1053] where points are constrained to be on $\Z^2$. Moreover, we numerically investigate the minimization problem for the energy per point among lattices for the Lennard-Jones potential $V=V_{\textnormal{LJ}}:r\mapsto r^{-12}-2r^{-6}$ as well as the Epstein zeta function associated to a $p$-norm $\|\cdot\|_p$, i.e., when $V=V_s:r\mapsto r^{-s}$, $s>2$. Our simulations show a new and unexpected phase transition for the minimizers with respect to $p$.
\end{abstract}


\medskip

\section{Introduction and setting}

Crystallization phenomena, i.e., the emergence of periodic crystal structures from physical forces, are at the same time very common in nature or experiments and known to be very difficult to justify mathematically \cite{RadinLowT,BlancLewin-2015,Canizo24}. What brings together particles in order to form periodic arrays -- like the ones composing the chair where probably sit the reader right now -- is believed to be a combination of potential and kinetic energies, atomic orbitals as well as screening or pressure effects. However, excepted in dimension 1 and 2 \cite{BlancLebris,Betermin:2014fy} and for very particular systems, no mathematical result incorporating all these parameters (e.g., quantum models) is available.

\medskip

Based on the number of degrees of freedom of such systems, it is indeed more reasonable to start with very simple models. Considering $N$ point particles $X_N=(x_1,...,x_N)\in (\R^n)^N$ interacting in pairs via a radial potential $V:\R_+\to \R$ is the most tractable model one can imagine. It basically corresponds to a zero temperature system with identical particle species. The total energy of such system is then usually given by
\begin{equation}\label{eq:E}
\textnormal{E}(X_N)=\frac{1}{2}\sum_{i=1 \atop i \neq j}^N\sum_{j=1}^N V(\|x_i-x_j\|_2),
\end{equation}
and it is then a key problem to show either that a minimizer of $\textnormal{E}$ is always a patch of a periodic structure (i.e., finite crystallization) or that the average energy per point for a minimizer converges, as $N\to +\infty$, to the energy per point of a lattice (i.e., crystallization in the sense of the thermodynamic limit).

\medskip

Whereas the problem is quite well-understood in dimension 1 \cite{VN1,VN2,Rad1}, it remains widely open in higher dimension. Crystallization results in the plane for energies of type $\textnormal{E}$ have all been shown for attractive-repulsive potentials. The Heitmann-Radin sticky disk potential $V_{\textnormal{HR}}$ defined by
\begin{equation}\label{eq:VHR}
\forall r\geq 0,\quad V_{\textnormal{HR}}(r)=\left\{\begin{array}{ll}
+\infty & \mbox{if $r\in[0,1)$}\\
-1 & \mbox{if $r=1$}\\
0 & \mbox{if $r>1$}
\end{array}
\right.
\end{equation}
is the prototypical example of such simple pair potential for which Heitmann and Radin have proved \cite{Rad2} (see also \cite{DelucaFriesecke-2018} for an alternative proof), mainly based on a result of Harborth \cite{harborth}, the minimality of a patch of the so-called triangular lattice of unit side-length (also called hexagonal or Abrikosov lattice \cite{Sandier_Serfaty}) given by
$$
\mathsf{A}_2:=\Z\left(1,0\right)\oplus \Z\left(\frac{1}{2},\frac{\sqrt{3}}{2} \right).
$$
The optimality of the triangular lattice has been also proved for other short-range \cite{Rad3,AuyeungFrieseckeSchmidt-2012,NinLucaSoft} and long-range \cite{Crystal} potentials. Very roughly, the proofs are always based on the fact that a regular hexagon is optimal locally (i.e., the maximal number of nearest-neighbours of a point is $6$) because of the large enough repulsion at short distance and that the tail of the potential is flat enough to almost not interfere in the perfect assembly of the hexagons.

\medskip

Recently, the same type of work -- with basically the same strategy -- has been done by the first author, De Luca and Petrache \cite{BDLPSquare}, for the soft potential $V_{\textnormal{BPD}}$ defined by
$$
\forall r\geq 0,\quad V_{\textnormal{BPD}}(r)=\left\{\begin{array}{ll}
+\infty & \mbox{if $r\in[0,1)$}\\
-1 & \mbox{if $r\in [1,\sqrt{2}]$}\\
0 & \mbox{if $r>\sqrt{2}$,}
\end{array}
\right.
$$
as well as some of its perturbations, for which the square lattice $\Z^2$ has been shown to be minimal, but only in the sense of the thermodynamic limit, for $\textnormal{E}$. No finite crystallization result of a patch of $\Z^2$ for $V_{\textnormal{BPD}}$ has been shown.

\medskip

The highly symmetric lattices $\mathsf{A}_2$ and $\Z^2$ are the too canonical good candidates for a crystallization problem \cite{Beterloc}, but other anisotropic structures obviously also appear as natural structures in the literature, see e.g., \cite{Lamy15,Ciftja22,Sheargridcells}. By a straightforward geometrical consideration, it is easy to show any crystallization on a given two-dimensional lattice $L$, by simply composing $V_{\textnormal{HR}}$ with an appropriate affine transform mapping $L$ to $\mathsf{A}_2$. The anisotropy obtained from such trick comes simply from the anisotropy of $L$ itself and seems then to be a little bit artificial. Notice also that algorithmic methods \cite{Torquato09} have been designed to construct interacting potential $V$ such that a given lattice is minimal.

\medskip

In this paper, in order to show the emergence of anisotropy in crystallization (i.e., the minimality of less symmetric lattices), we propose to investigate the minimizers of the following energy
\begin{equation}\label{eq:Enorm}
\textnormal{E}_{\|\cdot\|}(X_N)=\frac{1}{2}\sum_{i=1 \atop i \neq j}^N\sum_{j=1}^N V(\|x_i-x_j\|),
\end{equation}
where $\|\cdot\|$ is an arbitrary norm on $\R^2$, 
\begin{itemize}
\item[•] for the \textbf{Heitmann-Radin sticky disk potential} $V=V_{\textnormal{HR}}$ in the finite crystallization case. Using Brass' results \cite{Brass} on combinatorial geometry, generalizing Harborth's work \cite{harborth} to an arbitrary norm, as well as the recent crystallization result by Del Nin and De Luca in the square sticky disk case \cite{NinLuca25}, we completely classify the minimizers of $\textnormal{E}_{\|\cdot\|}$ according to the maximal number of nearest-neighbors (i.e., the kissing number), which is, up to an affine transform, $\mathsf{A}_2$ or $\Z^2$. We explicitly construct an infinite family of norms $\{N_{p,L}\}_{p\in [1,\infty]}$ in such a way that a given lattice $L$ is minimal for $\textnormal{E}_{N_{p,L}}$ and we investigate the minimizers for the $p$-norms. Furthermore, we give the minimal energy of $\textnormal{E}$ for $V_{\textnormal{BPD}}$ when the points are assumed to be already on $\Z^2$.
\item[•] for the \textbf{Lennard-Jones potential} $V_{\textnormal{LJ}}:\R_+^*\to \R$, $r\mapsto V_{\textnormal{LJ}}(r)=\displaystyle\frac{1}{r^{12}}-\frac{2}{r^6}$, among lattices (i.e., when the structure is already constrained to be periodic). After proving a result allowing us to consider only an associated (reduced) minimization problem among unit density lattices, we numerically investigate the minimizer of $\textnormal{E}_{\|\cdot\|_p}$ for different $p$-norms, $p\in [1,\infty]$. Furthermore, we also numerically study the same type of minimization problem among unit density lattices for the purely repulsive inverse power-law $V_s(r)=\frac{1}{r^s}$, $s>2$ -- i.e., the \textbf{Epstein zeta function} --, which is related to the Lennard-Jones one. Our numerical results show in both cases a new phase transition for the minimizer. Except for $p\in \{1,2,\infty\}$, the minimizer for $V_{\textnormal{LJ}}$ does not seem to be the same as the one for $V_{\textnormal{HR}}$. This shows a new phenomenon different from the one proven in the euclidean case in \cite{Rad2,LBComputerLJ23} for which a triangular lattice is optimal in both cases.
\end{itemize}

\noindent \textbf{Plan of the paper.} In Section \ref{sec:lattices}, we give some definitions and results concerning norms and lattices. Section \ref{sec:main} is devoted to the sticky disk potential with an arbitrary norm as well as $V_{\textnormal{BPL}}$ with points on $\Z^2$. Finally in Section \ref{sec:LJ}, we numerically and theoretically investigate the Lennard-Jones problem among lattices, as well as the Epstein zeta functions, with an arbitrary norm.

\section{Preliminaries about norms and lattices}\label{sec:lattices}

A key point in Brass' work is to classify norms with respect to their kissing number (see \cite[Section 2]{Swanepoel2018}), whose definition we recall here.

\begin{defi}[\textbf{Minimal-distance graph and kissing number for a norm}]
Let $\|\cdot\|$ be a norm on $\R^2$. For all configuration $X_N\subset \R^2$ such that $d(X_N):=\min_{i\neq j} \|x_i-x_j\|=1$, we define the minimal-distance graph $G_N=(X_N,E_N)$ where $X_N$ is its set of points and $E_N=\{\{x,y\}\subset X_N : \|x-y\|=1\}$ is its set of edges with unit length. The Hadwiger number, also called kissing number, associated to $\|\cdot\|$ is defined by
$$
K_{\|\cdot\|}:= \max_{X_N\subset (\R^2)^N \atop d(X_N)=1}\max_{x\in X_N} \deg(x)
$$
where $\deg(x)=\sharp\{y\in X_N : \{x,y\}\in E_N\}$ is the degree of $x$.
\end{defi}
\begin{example}
For instance, we know that $K_{\|\cdot\|_2}=6$ is the classical kissing number for euclidean balls in $\R^2$.
\end{example}

We have the following result that classifies norms with respect to their kissing number.

\begin{lemma}[\textbf{Classification of norms by their kissing number}, \cite{Grunbaum61}] \label{lem:kissing}
Let $\|\cdot\|$ be a norm on $\R^2$ and $S_{\|\cdot\|}=\{x\in \R^2 : \|x\|=1\}$ its unit sphere centered at the origin. We have: 
\begin{enumerate}
\item If $S_{\|\cdot\|}$ is a parallelogram, then $K_{\|\cdot\|}=8$,
\item If $S_{\|\cdot\|}$ is not a parallelogram, then $K_{\|\cdot\|}=6$.
\end{enumerate}
\end{lemma}

Therefore, we can define the following sets of norms with fixed kissing number:
$$
\forall k\in \{6,8\},\quad \mathcal{N}_k=\{\|\cdot\| \textnormal{ norm on $\R^2$} : K_{\|\cdot\|}=k\}.
$$

\begin{remark}[\textbf{Shape of the unit sphere and convexity}]
Notice that $S_{\|\cdot\|}$ is not a parallelogram if and only if $S_{\|\cdot\|}$ is strictly convex (see e.g., Brass \cite[Lemma 1]{Brass}).
\end{remark}

\begin{example}[\textbf{$p$-norms}]
We notice that $\|\cdot\|_1$ and $\|\cdot\|_\infty$ belongs to $\mathcal{N}_8$ -- since their unit spheres are squares -- whereas $\|\cdot\|_p$, $p\not\in \{1,\infty\}$, belongs to $\mathcal{N}_6$ -- since their unit spheres are strictly convex -- where the $p$-norms are defined, for all $x=(x_1,x_2)\in \R^2$,
$$
\forall p\geq 1, \|x\|_p=\left( |x_1|^p+|x_2|^p\right)^{\frac{1}{p}},\quad \|x\|_\infty=\max(|x_1|,|x_2|).
$$
\end{example}

Furthermore, we need to recall the notion of lattice since that will be a central object in our work.

\begin{defi}[\textbf{Lattice spaces}]
We say that $L\subset \R^2$ is a two-dimensional lattice if there exists a basis $(u,v)$ of $\R^2$ such that $L=\Z u \oplus \Z v$. We call $\mathcal{L}_2$ the set of two-dimensional lattices and $\mathcal{L}_2(1)\subset \mathcal{L}_2$ the set of unit density lattice, i.e., $L\in \mathcal{L}_2$ satisfying $|\det(u,v)|=1$.
\end{defi}
\begin{remark}[\textbf{Lattice reduction}]\label{rem:D}
Two important results coming from Geometry of Numbers will help us to simplify our methods:
\begin{enumerate}
\item Any $L\in \mathcal{L}_2$ can be written, up to isometry, in the reduced following form (see Engel \cite{Engel}):
\begin{equation}\label{eq:Engel}
L=\Z(u_1,0)\oplus \Z(v_1,v_2),\quad (u_1,v_1,v_2)\in \R_+^*\times \R_+\times \R_+^*, (\widehat{u,v})\in \left[\frac{\pi}{3},\frac{\pi}{2}  \right], \|(u_1,0)\|_2\leq \|(v_1,v_2)\|_2.
\end{equation}
\item Any $L\in \mathcal{L}_2(1)$ can be represented as a point $(x,y)\in \R^2$ belonging to the half-fundamental domain $\mathcal{D}$ defined by
\begin{equation} \label{eq:D}
\mathcal{D}:=\left\{(x,y)\in \R^2 : 0\leq x \leq \frac{1}{2}, y>0, x^2+y^2\geq 1\right\},
\end{equation}
in such a way that, up to rotation, we can write $L$ as in \eqref{eq:Engel},
\begin{equation}\label{eq:Lxy}
L=\Z\left( \frac{1}{\sqrt{y}},0\right)\oplus \Z\left( \frac{x}{\sqrt{y}}.\sqrt{y}\right).
\end{equation}
Notice that the square lattice $\Z^2$ and the triangular lattice $\sqrt{\frac{2}{\sqrt{3}}}\mathsf{A}_2$ of unit density are respectively parametrized in $\mathcal{D}$ by $(0,1)$ and $\left(\frac{1}{2},\frac{\sqrt{3}}{2}\right)$. \\
Moreover, is usual (see e.g., \cite{Mont}), by symmetry of the energies we consider, to use the above half-fundamental domain instead of the full one (see e.g., \cite[Appendix]{BFMaxTheta20}), where $x\in \left[ -\frac{1}{2},\frac{1}{2} \right]$.
\end{enumerate}
\end{remark}

We end this section with a very simple and useful result about linear transforms mapping a lattice to another in their reduced forms.

\begin{lemma}[\textbf{Linear map transforming a lattice into another}]\label{lem:linear}
Let $L',L\in \mathcal{L}_2$ written in their reduced form as in \eqref{eq:Engel}, i.e., $L'=\Z(u_1',0)\oplus \Z(v_1',v_2')$ and $L=\Z(u_1,0)\oplus \Z(v_1,v_2)$. There exists a unique linear transform $T_{L',L}:\R^2\to \R^2$ such that $T_{L',L}(L')=L$ which is defined by
$$
\forall (x,y)\in \R^2,\quad T_{L',L}(x,y)=\left(\frac{u_1}{u_1'}x+\frac{1}{v_2'}\left(v_1-\frac{u_1 v_1'}{u_1'} \right)y,\frac{v_2}{v_2'}y \right).
$$
\end{lemma}
\begin{proof}
Writing, for $(x,y)\in \R^2$, $T_{L',L}(x,y)=(ax+by,cx+dy)$ with $(a,b,c,d)\in \R^4$ and solving the two systems coming from the conditions $T_{L',L}(u_1',0)=(u_1,0)$ and $T_{L',L}(v_1',v_2')=(v_1,v_2)$ leads to the unique solution given in the statement.
\end{proof}

\section{Sticky disk potential and consequences}\label{sec:main}

The goal of this section is to show how Brass' results \cite{Brass} on configuration with maximal shortest distances gives a complete solution of the crystallization problem for $\textnormal{E}_{\|\cdot\|}$ defined by \eqref{eq:Enorm} where $\|\cdot\|$ is a arbitrary norm and $V=V_{\textnormal{HR}}$ is the Heitmann-Radin sticky disk potential defined by \eqref{eq:VHR}. Notice that the problem is invariant by translation. We write, for all $N\in \N\backslash \{0,1\}$,
$$
\mathcal{E}_{\|\cdot\|}(N)=\min_{X_N\subset \R^2} \textnormal{E}_{\|\cdot\|}(X_N).
$$

\subsection{Two key constructions}

We start to describe two key constructions of configurations with $N$ points on $\mathsf{A}_2$ and $\Z^2$ appearing respectively in \cite{Rad2} and \cite[Section 5]{Brass}. These will be our canonical minimizers.

\medskip

Let $N\in \N\backslash \{0,1\}$, then we define $H_N$ and $Z_N$ as follows:

\medskip

\noindent • \textbf{$H_N$ is an incomplete regular fully filled hexagon with $N$ points on $\mathsf{A}_2$} in the following sense. Let $s$ be the maximal integer such that $N\in \llbracket h_s,h_{s+1}-1\rrbracket $ where
$$
h_s:=3s^2+3s+1,
$$ 
then there exists a unique pair $(k,j)\in\llbracket 0,5\rrbracket\times\llbracket 0,s\rrbracket$ such that 
$$
N=3s^2+3s+1+(s+1)k+j,\quad s=\max\{\ell \in \N : N\in \llbracket h_\ell,h_{\ell+1}-1\rrbracket\},\quad k\in \llbracket 0,5\rrbracket,\quad j\in \llbracket 0,s\rrbracket.
$$ 
For a geometrical understanding of this decomposition, notice that $h_s$ is the $(s+1)$th centered hexagonal number, i.e., the number of points on $\mathsf{A}_2$ filling a perfect full regular hexagon $F_s$ with $s+1$ points on each 6 sides (see e.g., the set of blue dots in Figure \ref{fig:HN}). Therefore, $N-h_s=(s+1)k+j$ is the number of points filling the next layer of the ``next" regular hexagon where $k$ sides $B_{s,k}$ with $s+1$ points each are filled (in the direct sense or rotation) and $j$ is the number of the remaining points $R_{s,k,j}$. Therefore, such set of points can be written using complex numbers, defining $x_{m,n,s}^p:=e^{\frac{i p \pi}{3}}\left(m+ne^{\frac{i\pi}{3}} \right)$, as follows:
$$
H_N=F_s\cup B_{s,k} \cup R_{s,k,j},
$$
where
\begin{align*}
&F_s=\{x_{m,n,s}^p : (m,n,p)\in \N\times \N \times \llbracket 0,5 \rrbracket, m+n\leq s\},\\
&B_{s,k}=\{ x_{m,n,s}^r : (m,n,r)\in \N\times \N^*\times \llbracket 0,k-1 \rrbracket , m+n=s+1\},\\
&R_{s,k,j}=\{x_{m,n,s}^k : (m,n)\in \N \times \llbracket 1,j \rrbracket, m+n=s+1  \}.
\end{align*}
An illustration of this construction can be found on Figure \ref{fig:HN}. It has been shown in \cite{Rad2} that, for all $N\in \N\backslash \{0,1\}$, $H_N$ has exactly $\lfloor 3N-\sqrt{12N-3} \rfloor$ pairs of points with unit euclidean distances, which means that
$$
E_{\|\cdot\|_2}(H_N)=-\lfloor 3N-\sqrt{12N-3} \rfloor.
$$

\begin{figure}[!h]
\centering
 \includegraphics[width=6cm]{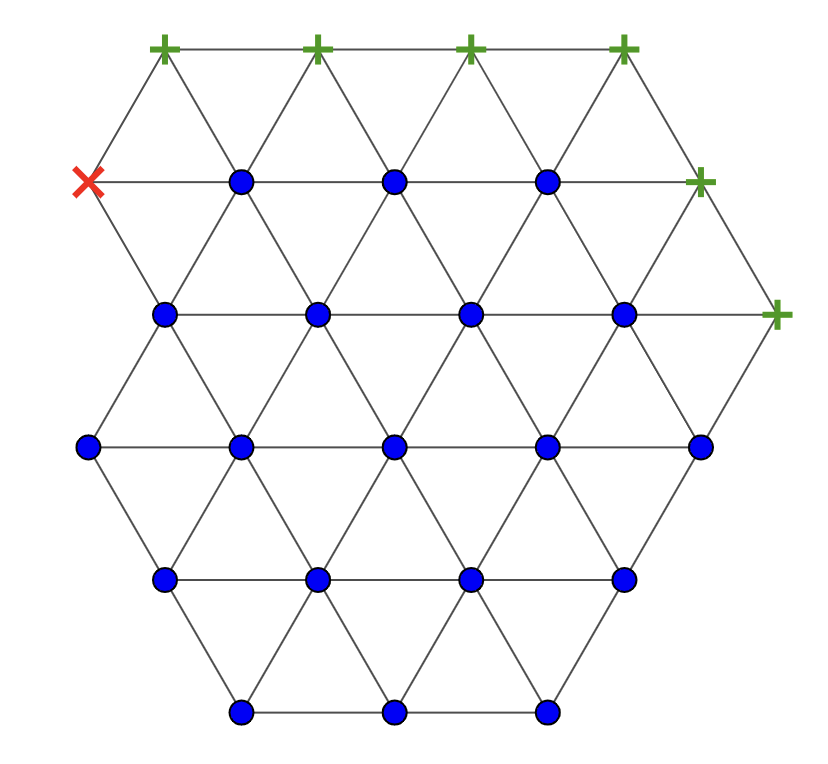}
\caption{Construction of the configuration $H_N$ for $N=26$. In this case, $(s,k,j)=(2,2,1)$ and we verify that the number of unit distances is indeed $\lfloor 3\times 26-\sqrt{12\times 26-3} \rfloor=60$. Blue dots \textcolor{blue}{\textbf{$\bullet$}} correspond to $F_{2}$, green crosses \textcolor{DarkGreen}{\textbf{+}} to $B_{2,2}$ and the red cross \textcolor{red}{\textbf{$\times$}} to $R_{2,2,1}$.}
\label{fig:HN}
\end{figure}

\medskip

\noindent • \textbf{$Z_N$ is an incomplete full filled octagon with $N$ points on $\Z^2$}. Contrary to the construction of $H_N$ given above which is rather easy to explain and for which we are able to give the exact coordinates of each point with a simple formula, the case of $Z_N$ is way more complicated. Therefore, we have chosen to give a similar explanation to the one given by Brass in \cite[p. 208-209]{Brass} with a bit more detail. First, let $N\in \N\backslash \{0,1\}$ and $s$ be the maximal integer such that $N\in \llbracket z_s,z_{s+1}-1\rrbracket $ where
$$
z_s:=7s^2+4s+1,
$$
then we call $i\in \llbracket 0, z_{s+1}-z_s-1 \rrbracket$ the integer such that 
$$
N=7s^2+4s+1+i.
$$ 
Notice that $z_s$ is the number of points on $\Z^2$ filling a perfect full regular (for $\|\cdot\|_\infty$) octagon $\tilde{F}_s$ with $s+1$ points on each of its $8$ sides. Let us call $\{A,B,C,D,E,F,G,H\}$ the sides of $\tilde{F}_s$, starting with the leftmost side and in clockwise direction (see Figure \ref{fig:ZN}). Thus, the $i$ next points (if $i\neq 0$) are added to the sides in such a way that a maximum of unit distances is created by each new point: i.e., in the order $H,A,B,C,D,B,E,F,D,C,D,E,F,G$ by always filling up a side until no further point of degree 4 can be added, before starting a new side (i.e., a new layer) by adding a point of degree 3. In Figure \ref{fig:ZN}, we have depicted how add points one by one (the points are labeled exactly as $N$) to $\tilde{F}_1$ and $\tilde{F}_2$ in order to fill entirely $\tilde{F}_3$, giving in one picture all the configurations $Z_N$ for $N\in \{2,...,z_3=76\}$.

\begin{figure}[!h]
\centering
 \includegraphics[width=10cm]{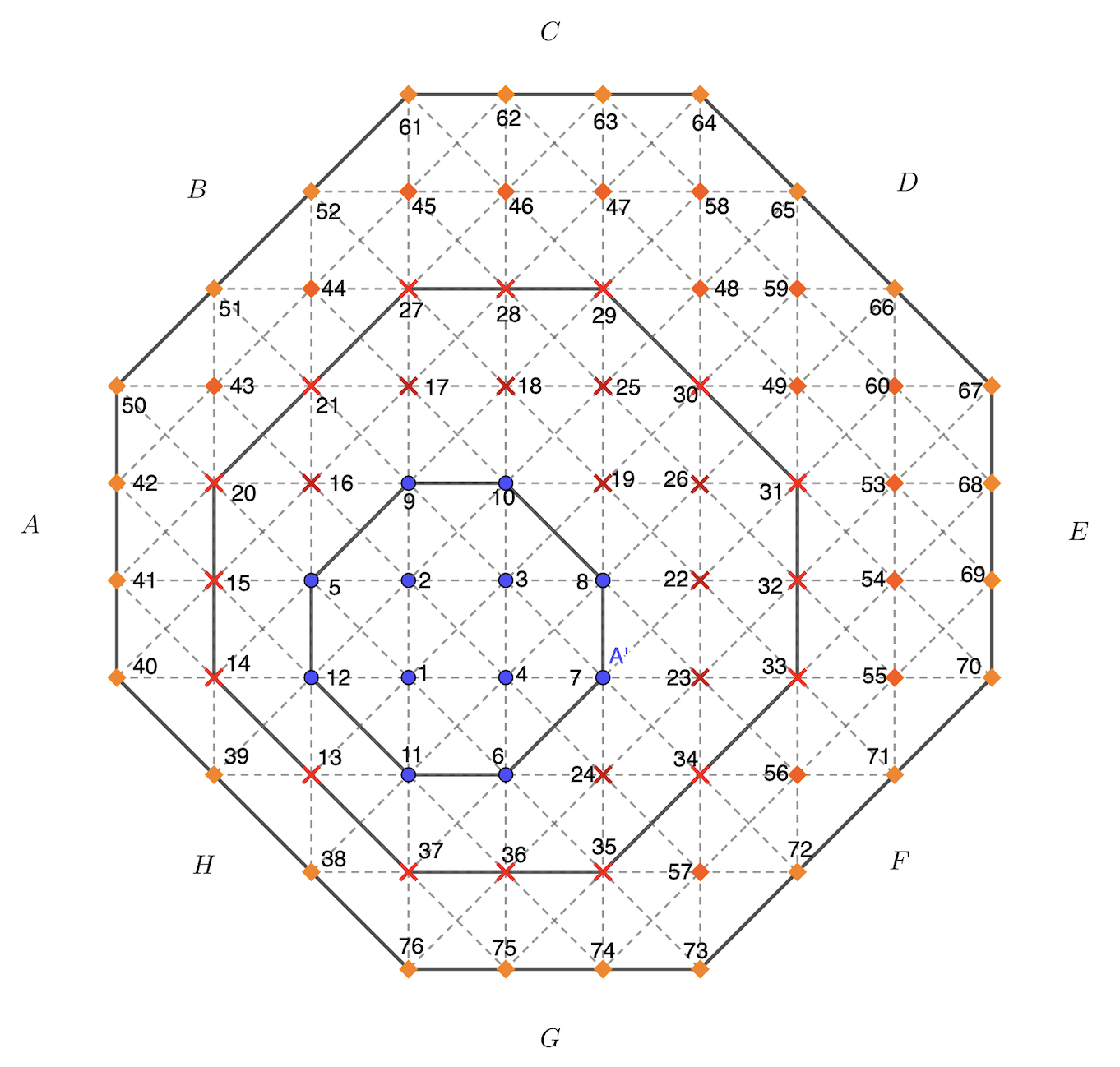}
\caption{Construction of the configuration $Z_N$ for $N\in \{2,...,z_3=76\}$. Solid lines materialize the boundary of the fully filled octagons $\tilde{F}_i$, $i\in \{2,3,4\}$ whereas dashed lines correspond to the other edges of the minimal-distance graph associated to $Z_N$ (for $\|\cdot\|_\infty$). The location of each new added points, to pass from $Z_N$ to $Z_{N+1}$, is indicated by the labels (i.e., the numbers correspond to the values of $N$). $\{A,B,C,D,E,F,G,H\}$ are the names of the sides (up to translation). Blue dots \textcolor{blue}{\textbf{$\bullet$}} corresponds to $\tilde{F}_2$, red crosses  \textcolor{red}{\textbf{$\times$}} are the points added to $\tilde{F}_2$ in order to fill $\tilde{F}_3$ and orange diamonds \textcolor{orange}{$\blacklozenge$} are the points added to $\tilde{F}_3$ in order to fill $\tilde{F}_4$.}
\label{fig:ZN}
\end{figure}

It has been shown in \cite{Brass} that, for all $N\in \N\backslash \{0,1\}$, $Z_N$ has exactly $\lfloor4N-\sqrt{28N-12} \rfloor$ pairs of points with unit distances for the infinite norm, which means that
$$
E_{\|\cdot\|_\infty}(Z_N)=-\lfloor4N-\sqrt{28N-12} \rfloor.
$$

\begin{remark}
We chose not to give the arguments explaining how to get the total number of edges for the minimal-distance graph of $H_N$ and $Z_N$ to keep our paper as short as possible. For the first one, details can be found in \cite{Rad2,FurlanettoTIPE}. For the second one, some details that we have carefully checked are given in \cite{Brass}.
\end{remark}

\subsection{Crystallization for Heitmann-Radin potential with arbitrary norm}

We can now state the following crystallization result for $\textnormal{E}_{\|\cdot\|}$ ensuring the existence of a lattice's patch as a minimizer and giving the precise value for $\mathcal{E}_{\|\cdot\|}(N)$ with respect to the norm's kissing number. 

\begin{thm}[\textbf{Crystallization for $E_{\|\cdot\|}$}]\label{thm:main}
Let $N\in \N\backslash \{0,1\}$. Let $\|\cdot\|$ be a norm on $\R^2$ and $\textnormal{E}_{\|\cdot\|}$ defined in \eqref{eq:Enorm} with $V=V_{\textnormal{HR}}$ the Heitmann-Radin sticky disk potential defined by \eqref{eq:VHR}. We have:
\begin{enumerate}
\item if $\|\cdot\|\in \mathcal{N}_6$, then 
$$
\mathcal{E}_{\|\cdot\|}(N)=-\lfloor 3N-\sqrt{12N-3} \rfloor ,
$$
achieved only for $X_N$ being a patch of a lattice $L_{\|\cdot\|}$ (up to isometry), in particular for $X_N=T_{\|\cdot\|}(H_N)$ where $T_{\|\cdot\|}\in \mathcal{L}(\R^2,\R^2)$ maps $\mathsf{A}_2$ to  $L_{\|\cdot\|}$.
\item if $\|\cdot\|\in \mathcal{N}_8$, then 
$$
\mathcal{E}_{\|\cdot\|}(N)=-\lfloor4N-\sqrt{28N-12} \rfloor,
$$
achieved only for $X_N$ being a patch of a lattice $L_{\|\cdot\|}$ (up to isometry), in particular for $X_N=T_{\|\cdot\|}(Z_N)$ where $T_{\|\cdot\|}\in \mathcal{L}(\R^2,\R^2)$. maps $\Z^2$ to $L_{\|\cdot\|}$.
\end{enumerate}
Furthermore, $L_{\|\cdot\|}$ is an asymptotic minimizer of $\textnormal{E}_{\|\cdot\|}$ in the sense of the thermodynamic limit, i.e.,
$$
\forall k\in \{6,8\},\quad \forall \|\cdot\|\in \mathcal{N}_k,\quad \lim_{N\to +\infty} \frac{\displaystyle\mathcal{E}_{\|\cdot\|}(N)}{N}=-\frac{k}{2}=\textnormal{E}^{\textnormal{per}}_{\|\cdot\|}\left(L_{\|\cdot\|}\right):=\frac{1}{2}\sum_{p\in L_{\|\cdot\|} \atop p\neq (0,0)} V_{\textnormal{HR}}(\|p\|),
$$
where $\textnormal{E}^{\textnormal{per}}_{\|\cdot\|}(L_{\|\cdot\|})$ is the energy per point of $L_{\|\cdot\|}$ for the interaction potential $V_{\textnormal{HR}}$.
\end{thm}
\begin{proof}
It is clear that minimizing $\textnormal{E}_{\|\cdot\|}$ is equivalent -- since the number of pairs at distance $1$ is multiplied by $-1$ in the energy -- to find a configuration $X_N$ of $N$ points in the plane such that $d(X_N)=1$, and with a maximal number of pairs at distance $1$. This is therefore the same as finding a configuration such that the maximum number of occurrences of the smallest distance in a set of $N$ points is achieved (up to a simple dilation). This problem has been studied by Brass in \cite[Theorem 3]{Brass} where it has been shown that if $S_{\|\cdot\|}$ is (resp. is not) a parallelogram, i.e., when $\|\cdot\|\in \mathcal{N}_8$ (resp. $\mathcal{N}_6$) then this number of occurrences is $\lfloor4N-\sqrt{28N-12} \rfloor$ (resp. $\lfloor 3N-\sqrt{12N-3} \rfloor $).\\
These two cases correspond to the two presented in Brass \cite[Section 5]{Brass} where the following has been shown:
\begin{enumerate}
\item if $\|\cdot\|\in \mathcal{N}_6$, then the euclidean norm is the prototypical case for which the minimizing configuration a subset of $\mathsf{A}_2$ achieved by the Heitmann-Radin's construction of $H_N$ we have recalled above. Indeed, for each $\|\cdot\|\in \mathcal{N}_6$ there is an affine image of $\mathsf{A}_2$ in which euclidean smallest distances are mapped on smallest distances of that norm. So each of the minimizers in the $\|\cdot\|_2$ case, which are known to be subsets of $\mathsf{A}_2$ as shown in \cite{Rad2}, corresponds to a set with the same number of smallest distances with respect to $\|\cdot\|$. Therefore, that ensures the existence of the linear map $T_{\|\cdot\|}$.
\item if $\|\cdot\|\in \mathcal{N}_8$, the same strategy is used for the infinite norm $\|\cdot\|_\infty$ instead of the euclidean one, which gives $\Z^2$ instead of $\mathsf{A}_2$ as shown by Brass, and for which $Z_N$ is optimal and ensures the existence of a linear map $T_{\|\cdot\|}$ such that a minimizer is given by $T_{\|\cdot\|}(Z_N)$.
\end{enumerate}
Furthermore, in both cases, any minimizer has to be a subset of $L_{\|\cdot\|}$, up to isometry, since it has been shown by Brass (and also Heitmann-Radin) for the triangular case and by Del Nin and De Luca \cite{NinLuca25} for the square case that minimizers are necessarily subsets of $\mathsf{A}_2$ for the euclidean norm and $\Z^2$ for the infinite norm. If there were a minimizer, for $\|\cdot\|\not\in \{ \|\cdot\|_2,\|\cdot\|_\infty\}$, being not a subset of the corresponding lattice $L_{\|\cdot\|}$, there would be, by linear transform, another minimizer being not a subset of $\mathsf{A}_2$ or $\Z^2$ in the euclidean or infinite norm case, which is therefore impossible.\\
The fact that $L_{\|\cdot\|}$ is asymptotically minimal in the sense of the thermodynamic limit is easy to show using the exact value of $\mathcal{E}_{\|\cdot\|}(N)$ (but also totally straightforward using the kissing number as in the proof of \cite[Theorem 2.1]{BDLPSquare}). The energy per point is therefore
$$
\frac{1}{2}\sum_{p\in L_{\|\cdot\|} \atop p\neq (0,0)} V_{\textnormal{HR}}(\|p\|)=\frac{1}{2}\sum_{p\in L_{\|\cdot\|} \atop \|p\|=1} -1=-\frac{1}{2}K_{\|\cdot\|},
$$
which completes the proof.
\end{proof}
\begin{remark}[\textbf{Lowest possible energy among norms and configurations}]
From this result, we automatically have
$$
\min_{\|\cdot\| \textnormal{ norm on $\R^2$}}\min_{X_N\subset \R^2} \textnormal{E}_{\|\cdot\|}(X_N)=-\lfloor4N-\sqrt{28N-12} \rfloor,
$$
achieved for any norm $\|\cdot\|\in \mathcal{N}_8$, i.e., a norm with a unit sphere being a parallelogram.
\end{remark}

\begin{remark}[\textbf{Uniqueness}]
It is important to notice the following: Theorem \ref{thm:main} does not claim that $T_{\|\cdot\|}(H_N)$ or $T_{\|\cdot\|}(Z_N)$ is the unique minimizer (up to isometry) of $\textnormal{E}_{\|\cdot\|}$. This fact can be only stated for $T_{\|\cdot\|}(H_N)$ and for special values of $N$, using a result by De Luca and Friesecke for the euclidean case \cite{DeLucaFrieseckeCassification}. Since we do not know the same for $Z_N$, we have chosen not to state any partial result of this kind.\\
Furthermore, let us mention that non-uniqueness for minimizers as well as the next-order term of order $N^{\frac{3}{4}}$ in the asymptotics of the energy as $N$ goes to infinity (including the Wullf shape) have been investigated, in the euclidean sticky disk or for other short-range interaction models, for instance in \cite{Schmidt2013,DPS17,MPSS19}.
\end{remark}

Furthermore, Lemma \ref{lem:linear} gives explicit formulas for $T_{\|\cdot\|}$ (once $L_{\|\cdot\|}$ is known).

\begin{corollary}[\textbf{Explicit linear maps}]
Let $L\in \mathcal{L}_2$ be written in its reduced form $L=\Z(u_1,0)\oplus \Z(v_1,v_2)$ as in \eqref{eq:Engel}, then:
\begin{enumerate}
\item There exists a unique linear transform $T_{\mathsf{A}_2,L}$ such that $T_{\mathsf{A}_2,L}(\mathsf{A}_2)=L$ which is given by $\forall (x,y)\in \R^2, T_{\mathsf{A}_2,L}(x,y)=\left(u_1 x + \frac{2}{\sqrt{3}}\left(v_1-\frac{u_1}{2}\right)y, \frac{2}{\sqrt{3}}v_2y \right)$.
\item There exists a unique linear transform $T_{\Z^2,L}$ such that $T_{\Z^2,L}(\Z^2)=L$ which is given by $\forall (x,y)\in \R^2,T_{\Z^2,L}(x,y)=(u_1x+v_1y,v_2y)$.
\end{enumerate}
\end{corollary}

\subsection{Consequences for the $p$-norms}

\begin{corollary}[\textbf{Crystallization for the $p$-norms}]\label{cor:pnorms}
Let $N\in \N\backslash \{0,1\}$. Let $\|\cdot\|_p$ be a $p$-norm, $p\geq 1$, and $\textnormal{E}_{\|\cdot\|_p}$ defined in \eqref{eq:Enorm} with $V=V_{\textnormal{HR}}$ the Heitmann-Radin sticky disk potential defined by \eqref{eq:VHR}. Then:
\begin{enumerate}
\item If $p=1$, then $\mathcal{E}_{\|\cdot\|_1}(N)=-\lfloor4N-\sqrt{28N-12} \rfloor$ achieved for $X_N=T_1(Z_N)$ where
$$
L_1=\Z(1,0)\oplus\Z\left(\frac{1}{2},\frac{1}{2} \right)\in \mathcal{L}_2,\quad \forall (x,y)\in \R^2, T_1(x,y)=\left(x+\frac{y}{2},\frac{y}{2}\right).
$$
\item If $p=\infty$, then $\mathcal{E}_{\|\cdot\|_\infty}(N)=-\lfloor4N-\sqrt{28N-12} \rfloor$ achieved for $X_N=Z_N\subset \Z^2$.
\item If $p>1$, $p\neq \infty$, then $\mathcal{E}_{\|\cdot\|_p}(N)=-\lfloor 3N-\sqrt{12N-3} \rfloor$ achieved for $X_N=T_p(H_N)\subset L_p$ where
$$
L_p:=\Z(1,0)\oplus \Z\left(\frac{1}{2},\left(1-\frac{1}{2^p}\right)^{\frac{1}{p}} \right)\in \mathcal{L}_2,\quad \forall (x,y)\in \R^2, T_p(x,y)=\left( x , \frac{2}{\sqrt{3}}\left(1-\frac{1}{2^p}\right)^{\frac{1}{p}}y\right).
$$
\end{enumerate}
\end{corollary}
\begin{proof}

First, it is known -- by a strict convexity argument -- that the unit sphere for $\|\cdot\|_p$, where $p\in [1,\infty]$, is a parallelogram if and only if $p\in \{1,\infty\}$. This gives, by Theorem \ref{thm:main} and Lemma \ref{lem:kissing}, the minimal energies as well as the fact that a minimizer is the image of $Z_N$ (for $p\in \{1,\infty\}$) or $H_N$ (for finite $p>1$) by a linear map. It remains to find the appropriate lattice as well as its associated linear map. \\
By simple computations of distances, the kissing numbers are achieved by:
\begin{itemize}
\item[•] the set of 8 points $\left\{(\pm 1,0), (0,\pm 1), \pm \left(\frac{1}{2},\frac{1}{2} \right),  \left(\pm\frac{1}{2},\mp\frac{1}{2} \right) \right\}$ on $S_{\|\cdot\|_1}$,
\item[•] the set of 8 points $\left\{ (\pm 1, 0), (0,\pm 1), \pm (1,1), (\pm 1, \mp 1) \right\}$ on $S_{\|\cdot\|_\infty}$,
\item[•] the set of 6 points $\left\{ (\pm 1, 0), \pm\left(\frac{1}{2},\left(1-\frac{1}{2^p}\right)^{\frac{1}{p}} \right), \left(\pm\frac{1}{2},\mp\left(1-\frac{1}{2^p}\right)^{\frac{1}{p}} \right) \right\}$ on $S_{\|\cdot\|_p}$, $p>1$.
\end{itemize}
Therefore, we automatically obtain the lattices $L_p$ and we use Lemma \ref{lem:linear} to get $T_p$ since $T_1=T_{\Z^2,L_1}$ and $T_p=T_{\mathsf{A}_2, L_p}$ where $p>1$.
\end{proof}

\begin{remark}[\textbf{Square versus octagon}]
In the $p=\infty$ case, even though any minimizer is a subset of $\Z^2$, the boundary of $Z_N$, when $N$ is a square number, is not a square but an octagon, which is due to the value of the kissing number (8). See for instance Figure \ref{fig:sqvsoct} where a perfect $4\times 4$ square and $Z_{16}$ are compared.\\

\begin{figure}[!h]
\centering
 \includegraphics[width=8cm]{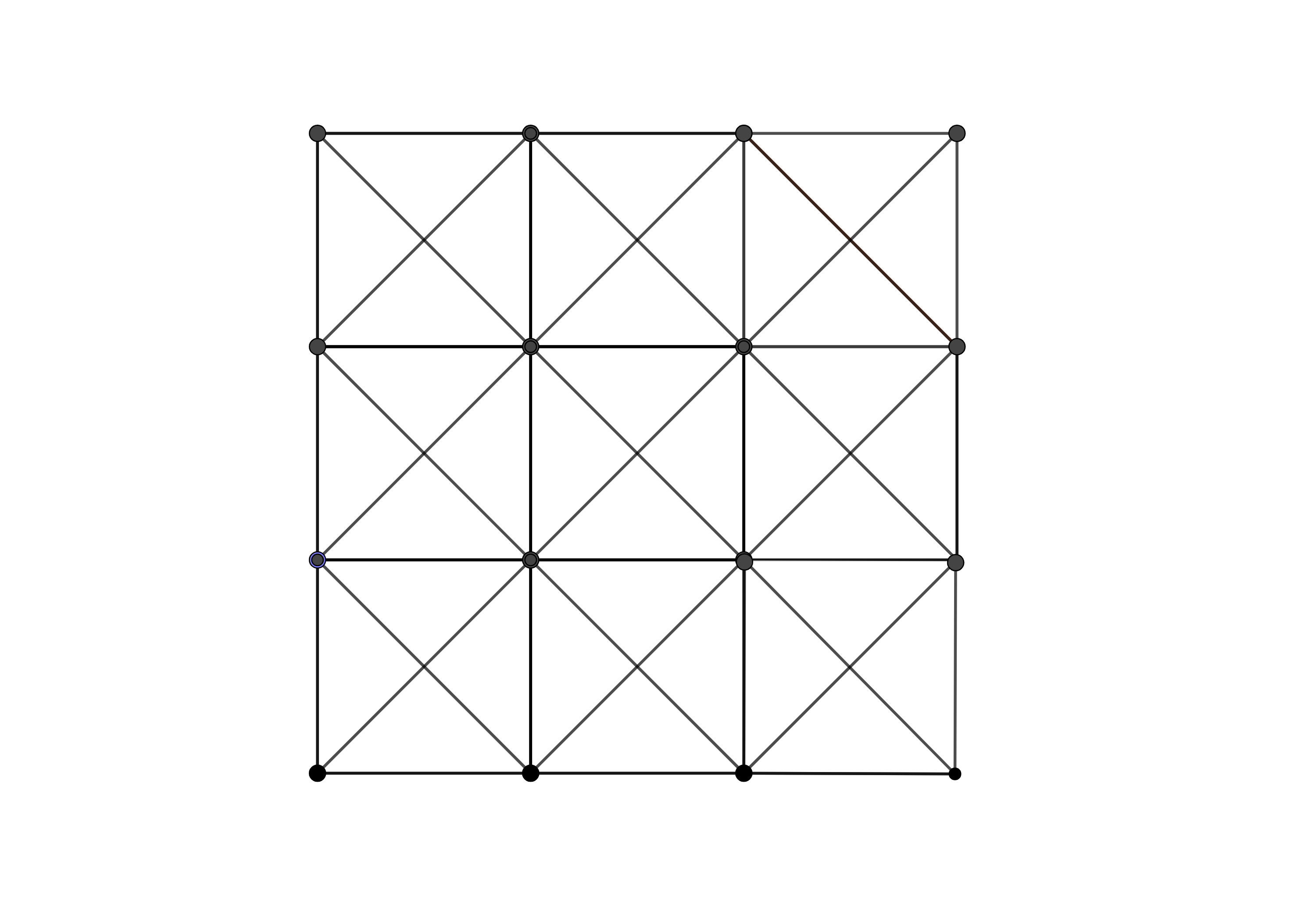} \quad 
\includegraphics[width=8cm]{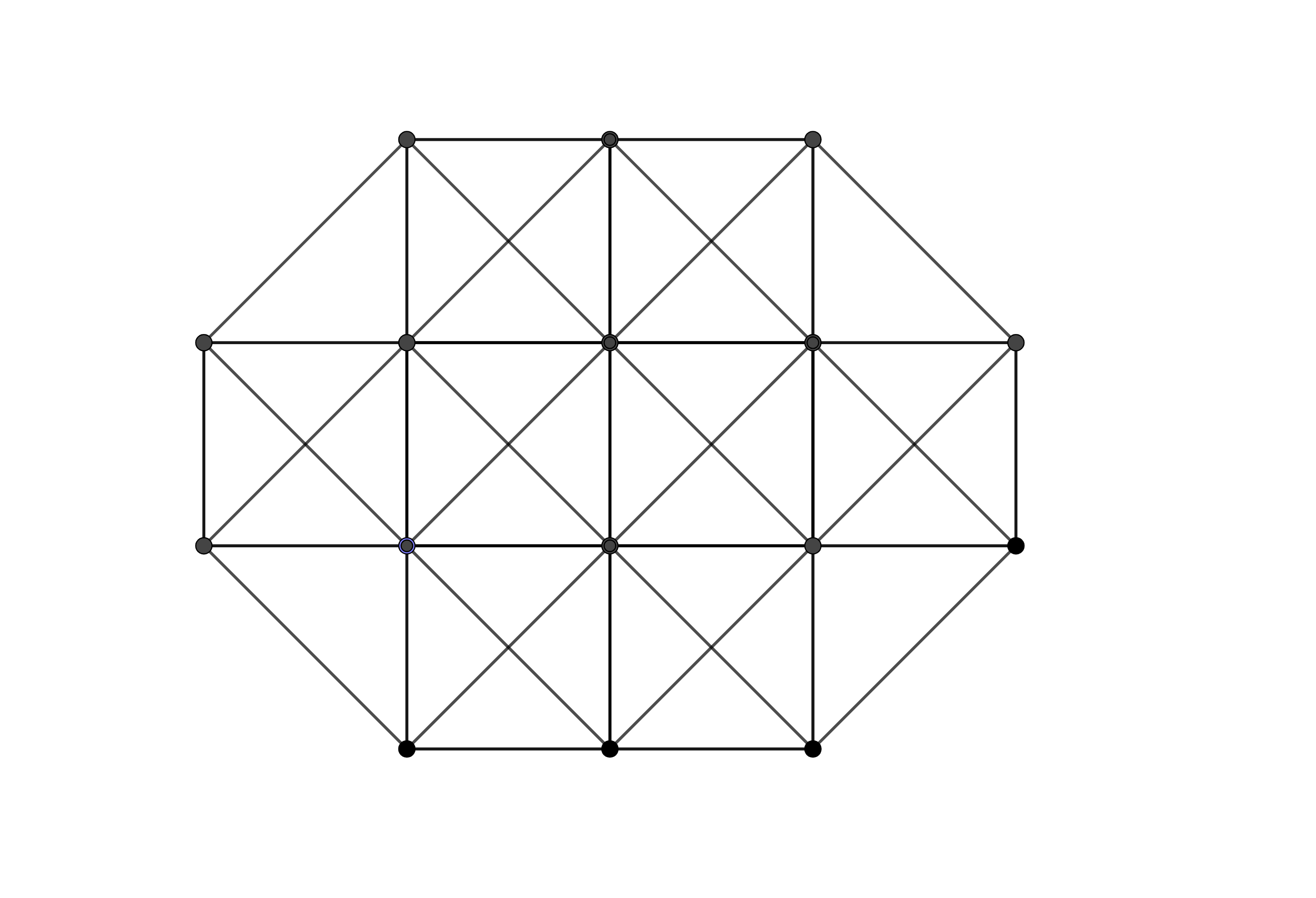}
\caption{For $N=16$ and the infinite norm $\|\cdot\|_\infty$, the perfect $4\times 4$ square configuration (on the left) has $42$ edges whereas the true minimal configuration $Z_{16}$ (on the right) has $43$ edges.}
\label{fig:sqvsoct}
\end{figure}
\end{remark}

\begin{example}
As examples, we have constructed the minimizers for the $p\in \{1,3\}$ cases in Figure \ref{fig:p13} below.
\begin{figure}[!h]
\centering
  \includegraphics[width=6cm]{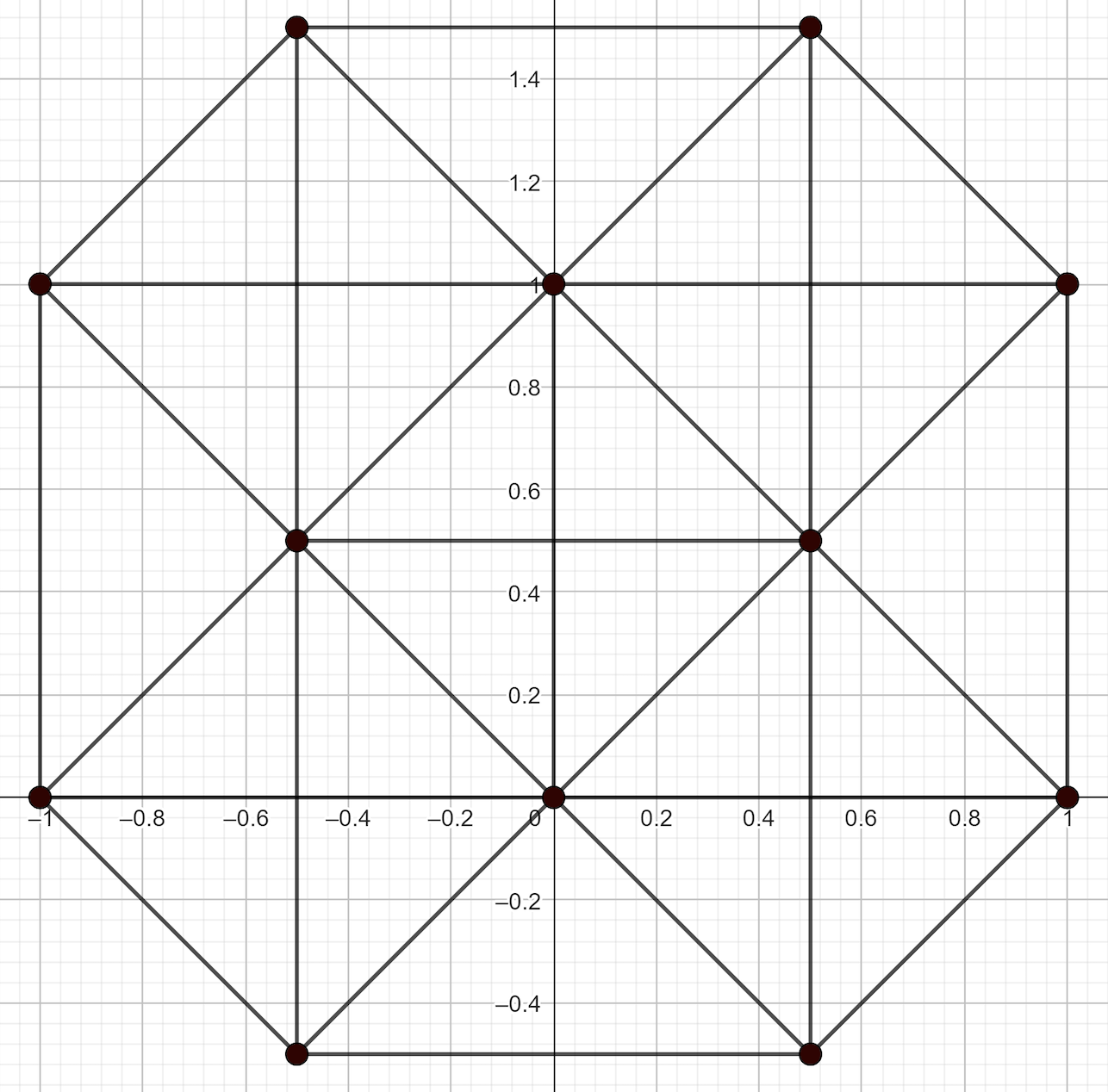} \quad\quad\quad \includegraphics[width=6cm]{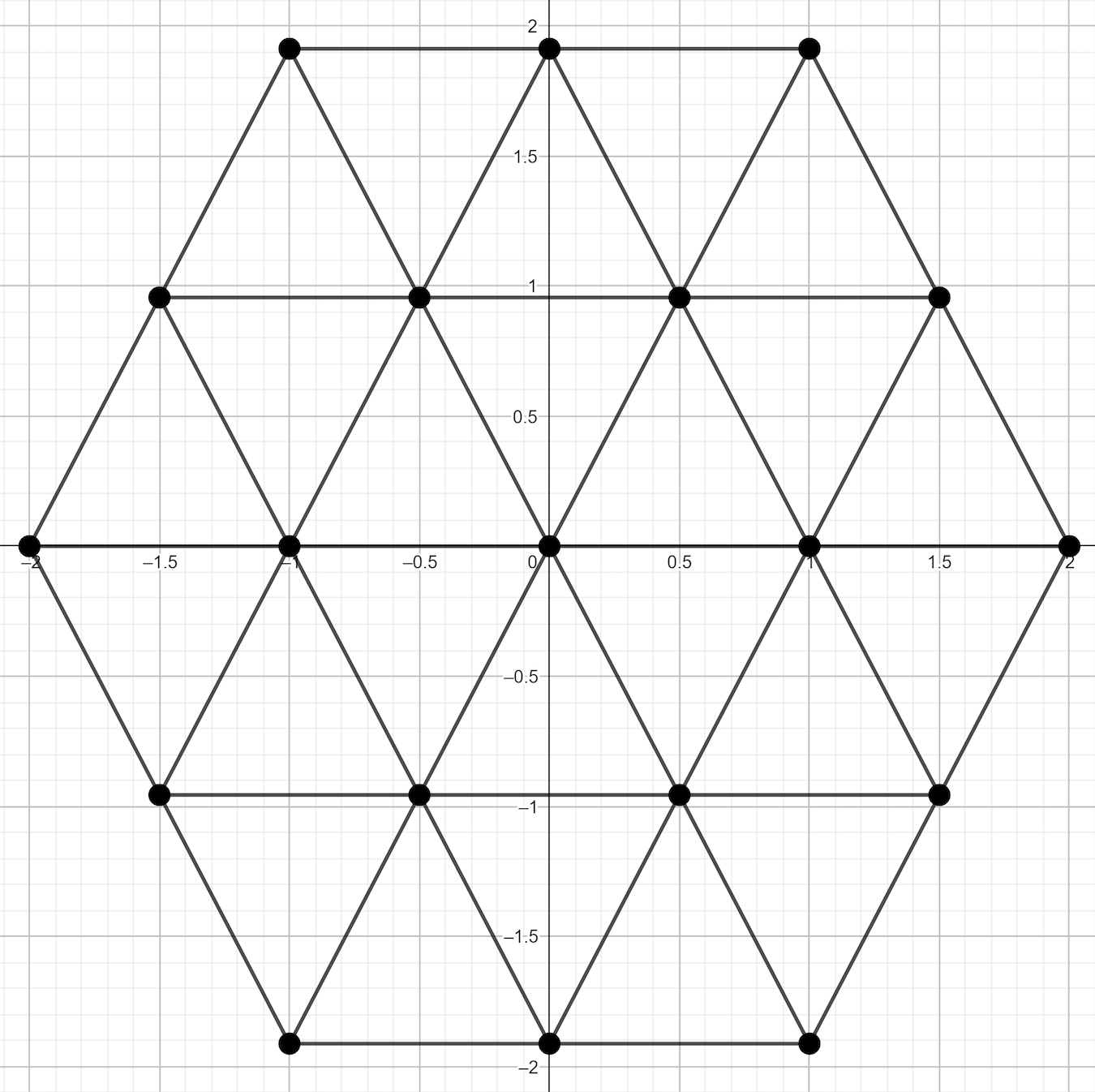}
\caption{Minimizer of $E_{\|\cdot\|_p}$ for $(p,N)=(1,12)$ (on the left) and $(p,N)=(3,19)$ (on the right)}
\label{fig:p13}
\end{figure}
\end{example}

The following result gives an infinite family of norms, based on $p$-norms, implying crystallization on a given lattice $L$.

\begin{prop}[\textbf{Norms implying crystallization on a given lattice}]\label{prop:normlinear}
Let $L\in \mathcal{L}_2$ written in its reduced form $L=\Z(u_1,0)\oplus \Z(v_1,v_2)$ as in \eqref{eq:Engel}, then for any $p\in [1,\infty]$, if $N_{p,L}:\R^2\to \R$ is defined by
$$
\forall (x_1,x_2)\in \R^2, N_{p,L}(x_1,x_2)=
 \left\{ \begin{array}{ll}
\sqrt[p]{\left| \frac{1}{u_1}x+\frac{1}{v_2}\left(\frac{1}{2}-\frac{v_1}{u_1} \right)y \right|^p + \displaystyle\left(1-\frac{1}{2^p} \right)\frac{|y|^p}{v_2^p}} & \mbox{if $p\geq 1$}\\
\max\left(\left|\frac{1}{u_1}x-\frac{v_1}{u_1 v_2}y \right|,\left|\frac{1}{v_2}y  \right| \right) & \mbox{if $p=\infty$,}
\end{array} \right.
$$
then $N_{p,L}$ is a norm and $\textnormal{E}_{N_{p,L}}$ is minimized by a patch of $L$, the latter being therefore an asymptotic minimizer in the thermodynamic limit sense for the energy.

\end{prop}
\begin{proof}
Let $p\in [1,\infty]$. With the notations of Corollary \ref{cor:pnorms}, let $T_{L_p,L}$ mapping $L_p$ to $L$. We look for $N_{p,L}$ of the form $N_{p,L}=\|\cdot\|_p\circ T$ with $T$ being an isomorphism from $\R^2$ to $\R^2$, which ensures the fact that $N_{p,L}$ will be a norm. Therefore, writing $S$ (resp. $S'$) the unit sphere for $N_{p,L}$ (resp. $\|\cdot\|_p$), we have, since $T_{L_p,L}(S')=S$ is unique,
$$
S=N_{p,L}^{-1}(\{1\})=(\|\cdot\|_p\circ T)^{-1}(\{1\}) \iff T(S)=\|\cdot\|_p^{-1}(\{1\})=S'\iff T=T_{L_p,L}^{-1}=T_{L,L_p}.
$$
We then easily compute, using Lemma \ref{lem:linear}, for all $(x,y)\in \R^2$ and $p\geq 1$, $p\neq \infty$,
$$
T_{L,L_p}(x,y)=\left(\frac{1}{u_1}x+\frac{1}{v_2}\left(\frac{1}{2}-\frac{v_1}{u_1} \right)y, \frac{1}{v_2}\left(1-\frac{1}{2^p} \right)^{\frac{1}{p}}y\right),\quad T_{L,L_\infty}(x,y)=\left( \frac{1}{u_1}x-\frac{v_1}{u_1v_2}y, \frac{1}{v_2}y\right),
$$
and we get the desired result by composition. Moreover, the asymptotic minimality of $L$ is a consequence of Theorem \ref{thm:main}.
\end{proof}
\begin{remark}[\textbf{Possible generalizations}]
We could obviously choose other norms by composing with the appropriate linear maps, e.g., writing all $N_{p,L}$'s expressions as $\|\cdot\|_\infty\circ T$ with a certain isomorphism $T$.
\end{remark}

\begin{example}[\textbf{Different ways to get a minimizing triangular patch}]
By Proposition \ref{prop:normlinear}, we see that our model can crystallize on a triangular lattice $\mathsf{A}_2$ in different ways. For instance:
\begin{itemize}
\item[•] by considering $\|\cdot\|_2$ for which the minimum of $\textnormal{E}_{\|\cdot\|_2}$ is directly a patch of $\mathsf{A}_2$ with the hexagonal basis of $H_N$ (see Figure \ref{fig:diffnorms}  on the left).
\item[•] by considering $N_{\infty,\mathsf{A}_2}$ for which the minimum of $\textnormal{E}_{N_{\infty,\mathsf{A}_2}}$ is also a patch of $\mathsf{A}_2$ but with a octagonal basis inherited from $Z_N$ (see Figure \ref{fig:diffnorms} on the right).
\end{itemize}
We therefore get two different configurations with a lower energy in the $N_{\infty,\mathsf{A}_2}$ case.

\begin{figure}[!h]
\centering
\includegraphics[width=8cm]{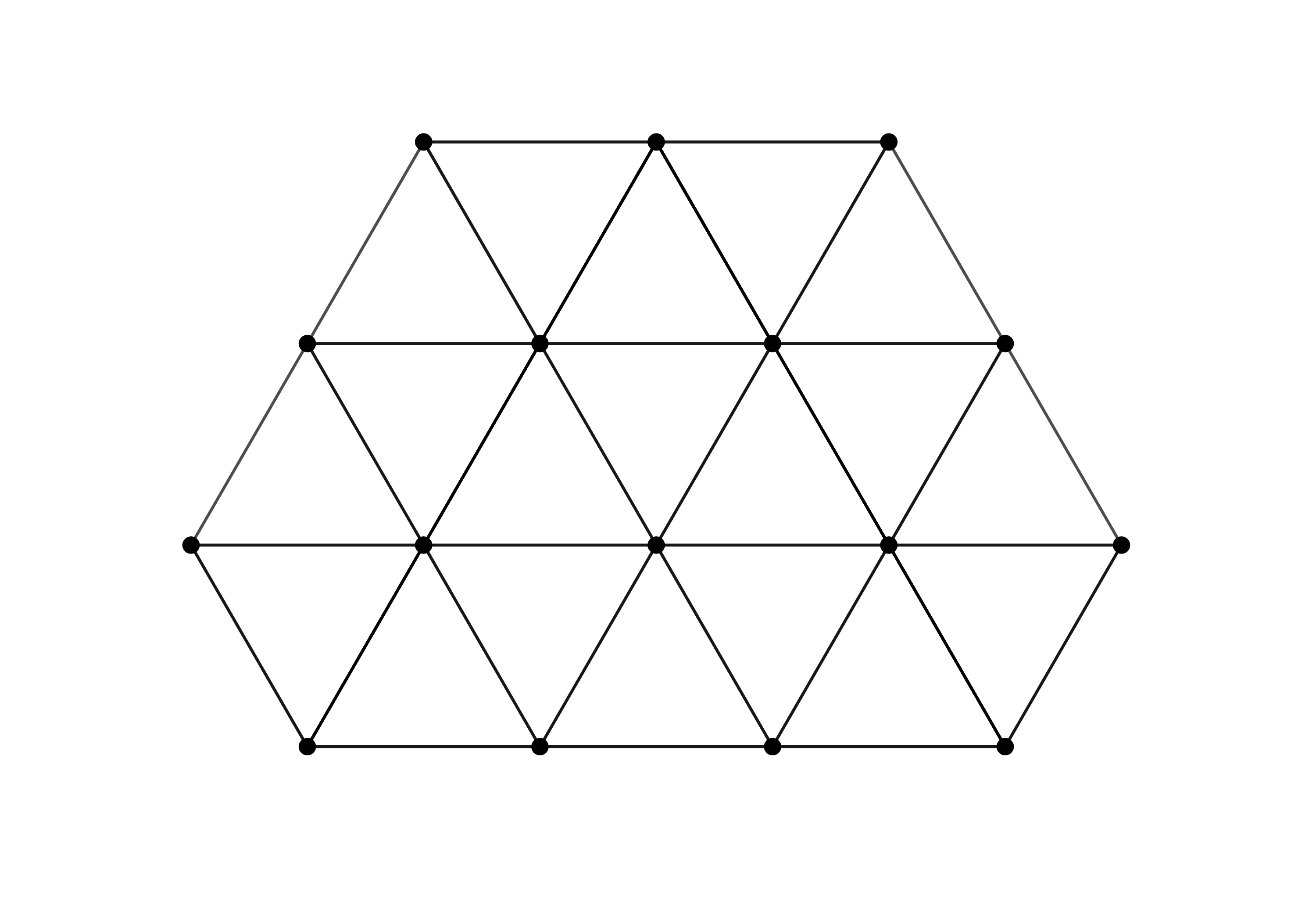}\quad \includegraphics[width=8cm]{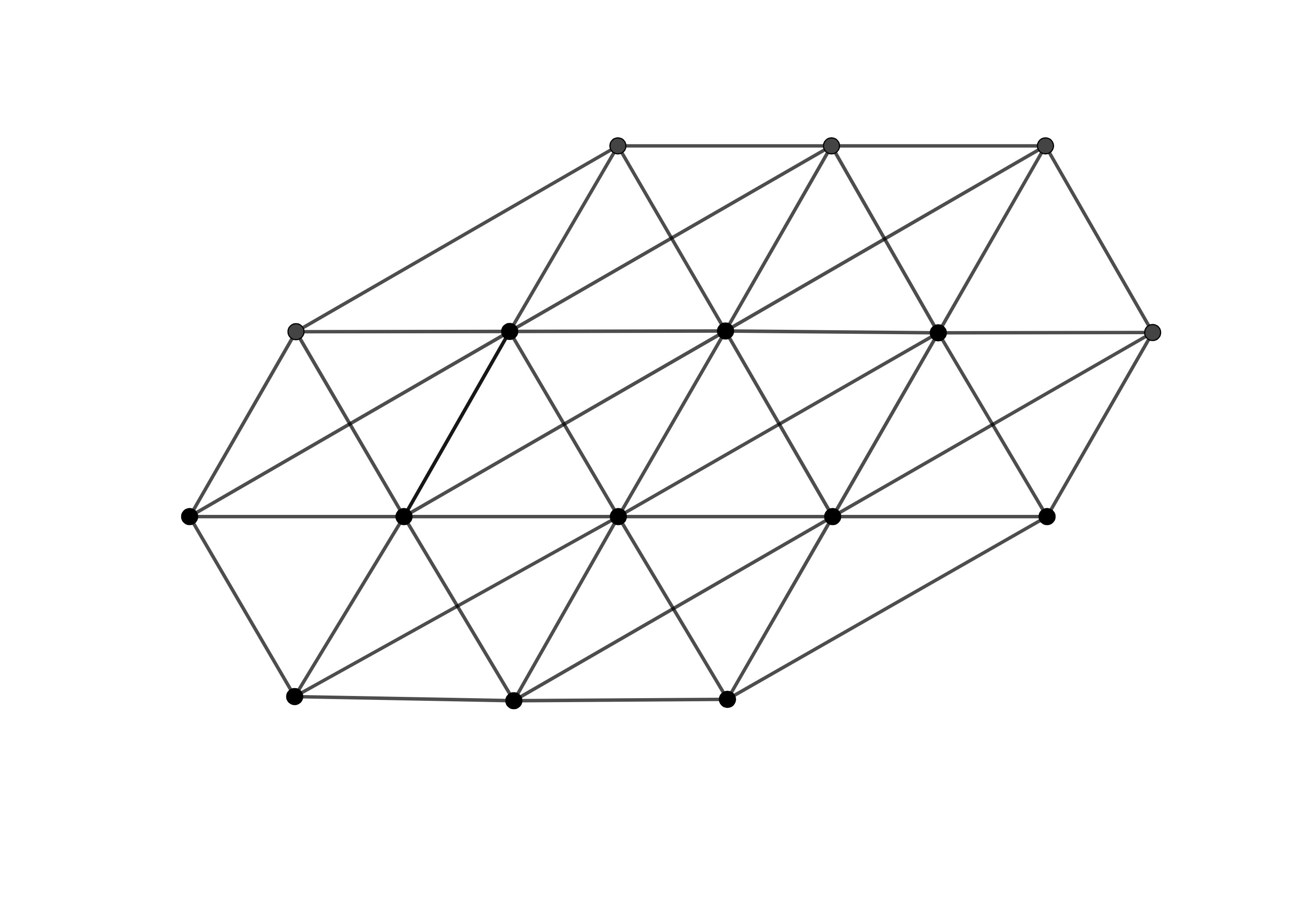}
\caption{Minimal configuration for $N=16$ for $\textnormal{E}_{\|\cdot\|_2}$ (on the left) and $\textnormal{E}_{N_{\infty,\mathsf{A}_2}}$ (on the right).}
\label{fig:diffnorms}
\end{figure}
\end{example}

\begin{remark}[\textbf{Open problems}]
In this paper, we do not investigate:
\begin{enumerate}
\item any short-range perturbation of $V_{\textnormal{HR}}$ as it was originally done by Radin in \cite{Rad3} or later by Au Yeung, Friesecke and Schmidt \cite{AuyeungFrieseckeSchmidt-2012} in the euclidean case. However, we strongly believe that Theorem \ref{thm:main} still holds -- but maybe not uniformly in $\|\cdot\|$ -- for perturbations of $V_{\textnormal{HR}}$ with strong repulsion before $1-\alpha$, a minimum equal to $-1$ at $r=1$, and the potential vanishing after $1+\beta$, with sufficiently small $\alpha,\beta>0$;
\item any properties of the Wulff shape, i.e., the minimizer of the macroscopic (continuous) energy corresponding to $\textnormal{E}_{\|\cdot\|}$ for fixed volume (see e.g., \cite{CicaleseLeonardi20}), for the potential $V_{\textnormal{HR}}$. This Wulff shape will probably be, in the $\|\cdot\|\in \mathcal{N}_6$ case, a linear transformation of the one found in \cite{AuyeungFrieseckeSchmidt-2012}, whereas the $\|\cdot\|\in \mathcal{N}_8$ case ask a thorough study and will probably yield a shape with octagonal symmetry.
\item any long-range perturbation of $V_{\textnormal{HR}}$ as it was done by Theil in \cite{Crystal} in the euclidean case. However, we again believe that an infinite family of long-range perturbations of $V_{\textnormal{HR}}$ can be constructed in such a way that an analogue of Theorem \ref{thm:main} holds.
\end{enumerate}
\end{remark}

\subsection{Connection with the soft potential $V_{\textnormal{BPD}}$ from \cite{BDLPSquare}}

As recalled in the introduction, the first author, de Luca and Petrache showed in \cite{BDLPSquare} the first crystallization results, in the sense of the thermodynamic limit, on $\Z^2$. Using the same strategy as in Theil's work \cite{Crystal}, the minimality of the square lattice for the two-body energy $\textnormal{E}_{\textnormal{BPD}}$ with potential $V=V_{\textnormal{BPD}}$ defined for $X_N\in (\R^2)^N$ by 
$$
\textnormal{E}_{\textnormal{BPD}}(X_N)=\frac{1}{2}\sum_{i=1 \atop i \neq j}^N\sum_{j=1}^N V_{\textnormal{BPD}}(\|x_i-x_j\|_2),\quad \forall r\geq 0, V_{\textnormal{BPD}}(r)=\left\{\begin{array}{ll}
+\infty & \mbox{if $r\in[0,1)$}\\
-1 & \mbox{if $r\in [1,\sqrt{2}]$}\\
0 & \mbox{if $r>\sqrt{2}$,}
\end{array}
\right.
$$
as well as some of its perturbations (short- and long-range) has been proven. In the finite $N$ case though, no crystallization result for this system is known. We can write the following result giving the minimal energy when the points are constrained to be on the square lattice.
\begin{prop}[\textbf{Minimality for $\textnormal{E}_{\textnormal{BPD}}$ on $\Z^2$}]
Let $N\in \N\backslash \{0,1\}$, then
$$
\min_{X_N\in (\Z^2)^N}\textnormal{E}_{\textnormal{BPD}}(X_N)=\mathcal{E}_{\|\cdot\|_\infty}(N)=-\lfloor4N-\sqrt{28N-12} \rfloor,
$$
achieved for $X_N=Z_N$.
\end{prop}
\begin{proof}
If we consider only points on $\Z^2$, it is clear that the pair distances giving a non-zero finite energy are exactly $1$ and $\sqrt{2}$, which allows us to say that minimizing $\textnormal{E}_{\textnormal{BPD}}$ on $\Z^2$ is equivalent to minimizing $\textnormal{E}_{\|\cdot\|_\infty}$ on $\Z^2$, which gives exactly the same minimizer and minimal energy as in the general case on $\R^2$. Therefore, using Corollary \ref{cor:pnorms}, the result is proved.
\end{proof}

Furthermore, it has been shown in \cite[Theorem 2.1]{BDLPSquare} that a vertex in the minimal-distance graph of a minimal configuration has maximum degree 8 and, if this number is achieved, then the configuration composed by these 9 points (the point and its 8 neighbors) is, up to an isometry, a subset of $\Z^2$. Since we believe that a minimizing configuration for $\textnormal{E}_{\textnormal{BPD}}$ has to have the maximum number of (interior) points with 8 neighbors, we conjecture that a minimizer (maybe not all of them?) must be a subset of $\Z^2$ and that the following statement should hold.

\begin{Conjecture}[\textbf{Minimality for $\textnormal{E}_{\textnormal{BPD}}$ in general}]
For all $N\in \N\backslash \{0,1\}$, we have
$$
\min_{X_N\in (\R^2)^N}\textnormal{E}_{\textnormal{BPD}}(X_N)=\mathcal{E}_{\|\cdot\|_\infty}(N)=-\lfloor4N-\sqrt{28N-12} \rfloor,
$$
achieved in particular for $X_N=Z_N$.
\end{Conjecture}

\begin{remark}[\textbf{Rigidity and minimizers}]
Let us add that, in contrast to the $\|\cdot\|_\infty$ case where there exists a certain rigidity (the energy selects only the lattice $\Z^2$, up to translation), the isotropy assumption in the BDP case, all rotations of $\Z^2$ are potential energy minimizers. One of the main challenges in that context lies in ruling out polycrystalline structures as energetically favorable as for instance in \cite{FKS2020} in the triangular lattice case.
\end{remark}

\section{Lennard-Jones problem among lattices with an arbitrary norm}\label{sec:LJ}
The goal of this part is to consider a classical example of attractive-repulsive potential widely used in molecular simulation and known to be a good approximation for illustrating some physical or chemical behaviors (see \cite[Section 6.3]{BetTheta15}): the Lennard-Jones potential 
$$
V_{\textnormal{LJ}}:\R_+^*\to \R, \quad r\mapsto V_{\textnormal{LJ}}(r)=\displaystyle\frac{1}{r^{12}}-\frac{2}{r^6}.
$$
In this paper, we do not cover the whole family of Lennard-Jones type potentials written as $r\mapsto \frac{a}{r^{p}}-\frac{b}{r^q}$, $(a,b)\in (\R_+^*)^2$, $p>q>2$ since the euclidean norm case seems to show a uniform behavior for the minimizer of the corresponding energy (see \cite{BetTheta15}). Contrary to the very simple attractive-repulsive Heitmann-Radin sticky disk potential, absolutely no rigorous crystallization result is available in the two-dimensional euclidean norm case, whereas many results are known in dimension one \cite{VN1,Rad1}. An interesting -- but still difficult -- object to study is the Lennard-Jones energy per point for lattices (i.e., the interaction between the origin and all the lattice points) for which many results are known \cite{BetTheta15,LBbonds21,BDefects20,LuoWei22,LBComputerLJ23,SunWeiZou24} in the euclidean case.

\medskip

For all $L\in \mathcal{L}_2$, we define
$$
\textnormal{E}_{\|\cdot\|}^{\textnormal{LJ}}[L]:=\sum_{p\in L\backslash \{0\}}\left(\frac{1}{\|p\|^{12}}-\frac{2}{\|p\|^6}  \right)=\zeta_{L,\|\cdot\|}(12)-2\zeta_{L,\|\cdot\|}(6),
$$
where the Epstein zeta function, with parameter $s$, associated to the norm $\|\cdot\|$ is defined by
$$
\forall s>2,\quad \zeta_{L,\|\cdot\|}(s):=\sum_{p\in L\backslash \{0\}} \frac{1}{\|p\|^s},
$$
which is again an energy per point, associated to the potential $V_{s}:\R_+^*\to \R$, $r\mapsto V_{s}(r)=\frac{1}{r^s}$.\\
Notice that if, for all $N\in \N^*$, $L_N\subset L$ is a lattice patch with $N$ points such that $L_N\to L$ as $N\to +\infty$ (for instance $L_N=K_N\cap L$ with a sequence of compact sets $K_N$ such that $\sharp L_N=N$ and $K_N\to \R^2$), then, respectively for $V=V_{\textnormal{LJ}}$ and $V=V_s$ with $\textnormal{E}_{\|\cdot\|}$ defined again by \eqref{eq:Enorm},
$$
\lim_{N\to +\infty} \frac{E_{\|\cdot\|}(L_N)}{N}=2\textnormal{E}_{\|\cdot\|}^{\textnormal{LJ}}[L],\quad  \lim_{N\to +\infty} \frac{E_{\|\cdot\|}(L_N)}{N}=2\zeta_{L,\|\cdot\|}(s).
$$
These double sums are absolutely convergent -- by equivalences of all norms in $\R^2$ and since this is true for the euclidean one. Furthermore, the Epstein zeta function satisfies the following homogeneity formula:
\begin{equation}\label{eq:zetahom}
\forall s>2,\quad \forall \lambda>0,\quad\zeta_{\lambda L,\|\cdot\|}(s)=\lambda^{-s}\zeta_{L,\|\cdot\|}(s).
\end{equation}

The following result, mainly based on \eqref{eq:zetahom}, is similar to the one shown in \cite[Theorem 2.6]{LBComputerLJ23} (see also \cite[Thm. 1.11]{OptinonCM}) and can be generalized to any Lennard-Jones type potential. It allows us to simplify the problem of minimizing $\textnormal{E}_{\|\cdot\|}^{\textnormal{LJ}}$ on $\mathcal{L}_2$ into the minimization of another lattice energy on the smaller space $\mathcal{L}_2(1)$.

\begin{prop}[\textbf{Minimal Lennard-Jones energy among dilated lattices}] \label{prop:LJ}
 For all $L\in \mathcal{L}_2(1)$, 
$$
 e_{\|\cdot\|}(L):=\min_{\lambda>0} \textnormal{E}_{\|\cdot\|}^{\textnormal{LJ}}[\lambda L]=\textnormal{E}_{\|\cdot\|}^{\textnormal{LJ}}[\lambda_L L]=-\frac{\zeta_{ L,\|\cdot\|}^2(6)}{\zeta_{L,\|\cdot\|}(12)}, \quad \textnormal{where}\quad \lambda_L:=\left(\frac{\zeta_{L,\|\cdot\|}(12)}{\zeta_{L,\|\cdot\|}(6)} \right)^{\frac{1}{6}}.
$$
\end{prop}
\begin{proof}
For all $\lambda>0$, and all $L\in \mathcal{L}_2(1)$, \eqref{eq:zetahom} implies that
$$
\textnormal{E}_{\|\cdot\|}^{\textnormal{LJ}}[\lambda L]=\zeta_{\lambda L,\|\cdot\|}(12)-2\zeta_{\lambda L,\|\cdot\|}(6)=\lambda^{-12}\zeta_{L,\|\cdot\|}(12)-2\lambda^{-6}\zeta_{L,\|\cdot\|}(6),
$$
and therefore,
$$
\frac{d}{d\lambda }\textnormal{E}_{\|\cdot\|}^{\textnormal{LJ}}[\lambda L]\geq 0 \iff -12\lambda^{-13}\zeta_{L,\|\cdot\|}(12)+12\lambda^{-7}\zeta_{L,\|\cdot\|}(6)\geq 0 \iff \lambda\geq \lambda_L=\left(\frac{\zeta_{L,\|\cdot\|}(12)}{\zeta_{L,\|\cdot\|}(6)} \right)^{\frac{1}{6}}.
$$
This proves that the minimum of $\lambda\mapsto \textnormal{E}_{\|\cdot\|}^{\textnormal{LJ}}[\lambda L]$ is achieved at $\lambda=\lambda_L$. The value of this minimum is given by 
\begin{align*}
 e_{\|\cdot\|}(L)=\textnormal{E}_{\|\cdot\|}^{\textnormal{LJ}}[\lambda_L L]&=\left(\frac{\zeta_{L,\|\cdot\|}(12)}{\zeta_{L,\|\cdot\|}(6)} \right)^{-2}\zeta_{ L,\|\cdot\|}(12)-2\left(\frac{\zeta_{L,\|\cdot\|}(12)}{\zeta_{L,\|\cdot\|}(6)} \right)^{-1}\zeta_{ L,\|\cdot\|}(6)=-\frac{\zeta_{ L,\|\cdot\|}^2(6)}{\zeta_{L,\|\cdot\|}(12)}.
\end{align*}
\end{proof}

In the euclidean norm case, the methods to find the minimizer $\lambda_{\mathsf{A}_2} \mathsf{A}_2$ of the Lennard-Jones energy in $\mathcal{L}_2$ are all based on the minimality of $\mathsf{A}_2$ in $\mathcal{L}_2(1)$ proved in \cite{Rankin,Cassels,Eno2,Diananda,Mont} for $\zeta_{L,\|\cdot\|_2}(s)$ for all $s>2$ (see \cite{LuoWei22,LBComputerLJ23}). In particular, the first author proved in \cite{LBComputerLJ23} the optimality of $\sqrt{\frac{2}{\sqrt{3}}}\mathsf{A}_2$ for $e_{\|\cdot\|_2}$, using a computer-assisted proof which actually allows us to generalize this result to other chosen exponents (rather than 12 and 6). For the general norm case, it is unclear what lattice minimizes $\zeta_{L,\|\cdot\|}(s)$ in $\mathcal{L}_2(1)$. As a first incursion in this new minimization problem, we propose a very simple numerical study of the Epstein zeta function $\zeta_{L,\|\cdot\|_p}(s)$, $s\in \{6,12\}$, associated to some of the $p$-norms $\|\cdot\|_p$, on $\mathcal{L}_2(1)$, as well as the same question for the Lennard-Jones energy $\textnormal{E}_{\|\cdot\|_p}^{\textnormal{LJ}}$ in $\mathcal{L}_2$, which is the same, by Proposition \ref{prop:LJ}, as minimizing $e_{\|\cdot\|_p}$ in $\mathcal{L}_2(1)$.

\medskip

We therefore parametrize any $L\in \mathcal{L}_2(1)$ by $(x,y)$ in the half-fundamental domain $\mathcal{D}$ defined by \eqref{eq:D} which allows us, via \eqref{eq:Lxy}, to write, for any norm $\|\cdot\|$ on $\R^2$ and any $s>2$,
$$
\zeta_{L,\|\cdot\|}(s)=\sum_{m,n} \frac{1}{\displaystyle\left\|\left( \frac{m+nx}{\sqrt{y}}, n\sqrt{y}\right)  \right\|^s},\quad e_{\|\cdot\|}(L):=\min_{\lambda>0} \textnormal{E}_{\|\cdot\|}^{\textnormal{LJ}}[\lambda L]=-\frac{\left(\displaystyle\sum_{m,n} \left\|\left( \frac{m+nx}{\sqrt{y}}, n\sqrt{y}\right)  \right\|^{-6} \right)^2}{\displaystyle\sum_{m,n} \left\|\left( \frac{m+nx}{\sqrt{y}}, n\sqrt{y}\right)  \right\|^{-12}},
$$
where all these summations are for $(m,n)\in \Z^2\backslash \{(0,0)\}$. We recall (see Remark \ref{rem:D}) that in this parametrization, the square lattice $\Z^2$ is parametrized by $(0,1)$ and the triangular lattice $\sqrt{\frac{2}{\sqrt{3}}}\mathsf{A}_2$ (where the coefficient $\sqrt{\frac{2}{\sqrt{3}}}$ ensures the density to be one) is parametrized by $\left(\frac{1}{2},\frac{\sqrt{3}}{2}\right)\approx(0.5,0.8660254)$ in $\mathcal{D}$.
\medskip

Using at the same time the Nelder-Mead minimization algorithm and the contour plots (in order to avoid finding only local minima) -- see for instance Figure \ref{fig:contour} for examples -- of the above functions, we got the following observations.

\medskip

• \textbf{Numerical results for the Epstein zeta function.} For the problem of minimizing $L\mapsto \zeta_{L,\|\cdot\|_p}(s)$ in $\mathcal{L}_2(1)$ for $s\in \{6,12\}$ we found the existence of $p_2>p_1>1$ such that:
\begin{itemize}
\item[-] if $p\in [1,2]$, the minimizer is $\sqrt{\frac{2}{\sqrt{3}}}\mathsf{A}_2$,
\item[-]  if $p\in (2,p_1)$, the minimizer is parametrized by a point of the form $(1/2,y_p)$ where $p\mapsto y_p$ is increasing from $\sqrt{3}/2$ to $\approx 0.97$,
\item[-]  if $p\in (p_1,p_2)$, the minimizer is $\Z^2$,
\item[-]  if  $p>p_2$, the minimizer is parametrized by a point of the form  $(1/2,y_p)$ where $p\mapsto y_p$ is decreasing from $\approx 1.097$ and going to $y_\infty\approx 1.095$ as $p\to +\infty$.
\end{itemize}
The value of $(p_1,p_2)$ belongs to $(2.4,2.5)\times (5.4,5.5)$ for $s=6$ and to $(2.7,2.8)\times (8.9, 9)$ for $s=12$.\\

The above phase transition of the minimizer with respect to the $p$-norm is at the same time intriguing and new even though other numerical \cite{C0SM01205J} and theoretical \cite[Prop. 3.4]{BetTheta15} results have been given in the past concerning the optimality of non-triangular lattices for purely repulsive convex potentials. Furthermore, except for $p=2$, it seems that $L_p$ does not play any role in this minimization process. We find the triangular phase $p\in [1,2)$ as well as the square one $p\in (p_1,p_2)$ really unexpected.

\medskip

• \textbf{Numerical results for the reduced Lennard-Jones energy.} For the problem of minimizing $L\mapsto e_{\|\cdot\|_p}(L)$ in $\mathcal{L}_2(1)$, we found the existence of $p_1\in (1.2,1.3)$, $p_2\in (3.6,3.7)$ and $p_3\in (3.8,3.9)$ such that:
\begin{enumerate}
\item[-]  if $p=1$, the minimizer is $\Z^2$,
\item[-]  if $p\in (1,p_1)$, the minimizer is parametrized by a point of the form $(\varepsilon_p,\sqrt{1-\varepsilon_p^2})$ for very small $\varepsilon_p$ such that $p\mapsto \varepsilon_p$ increases,
\item[-]  if $p\in (p_1,2]$, the minimizer is $\sqrt{\frac{2}{\sqrt{3}}}\mathsf{A}_2$,
\item[-]  if $p\in (2,p_2)$, the minimizer is parametrized by a point of the form $(1/2,y_p)$ where $p\mapsto y_p$ is increasing from $\sqrt{3}/2$ to $\approx 0.99$,
\item[-]  if $p\in (p_2,p_3)$, the minimizer is parametrized by a point of the form  $(x_p, y_p)$ where $p\mapsto x_p$ is decreasing from $1/2$ to $0$ and $p\mapsto y_p$ is increasing from $0.99$ to $1$,
\item[-]  $p\in (p_3,\infty]$, the minimizer is $\Z^2$.
\end{enumerate}

This type of phase transition with respect to the norm is again new and interesting. Whereas the $p\in \{1,2,\infty\}$ cases are similar to what we showed in Corollary \ref{cor:pnorms}, i.e., a multiple of $L_p$ is optimal, the other cases are quite surprising. For the latter, the local and global effect of the Lennard-Jones potential seems to perturb the local effect of the norm (given by Corollary \ref{cor:pnorms}). Notice that an attractive-repulsive potential having not a triangular lattice as global minimizer (but a square one) has been designed by the first author and Petrache in \cite[Prop. 1.17]{OptinonCM}.

\medskip

\begin{figure}[!h]
\centering
\includegraphics[width=8cm]{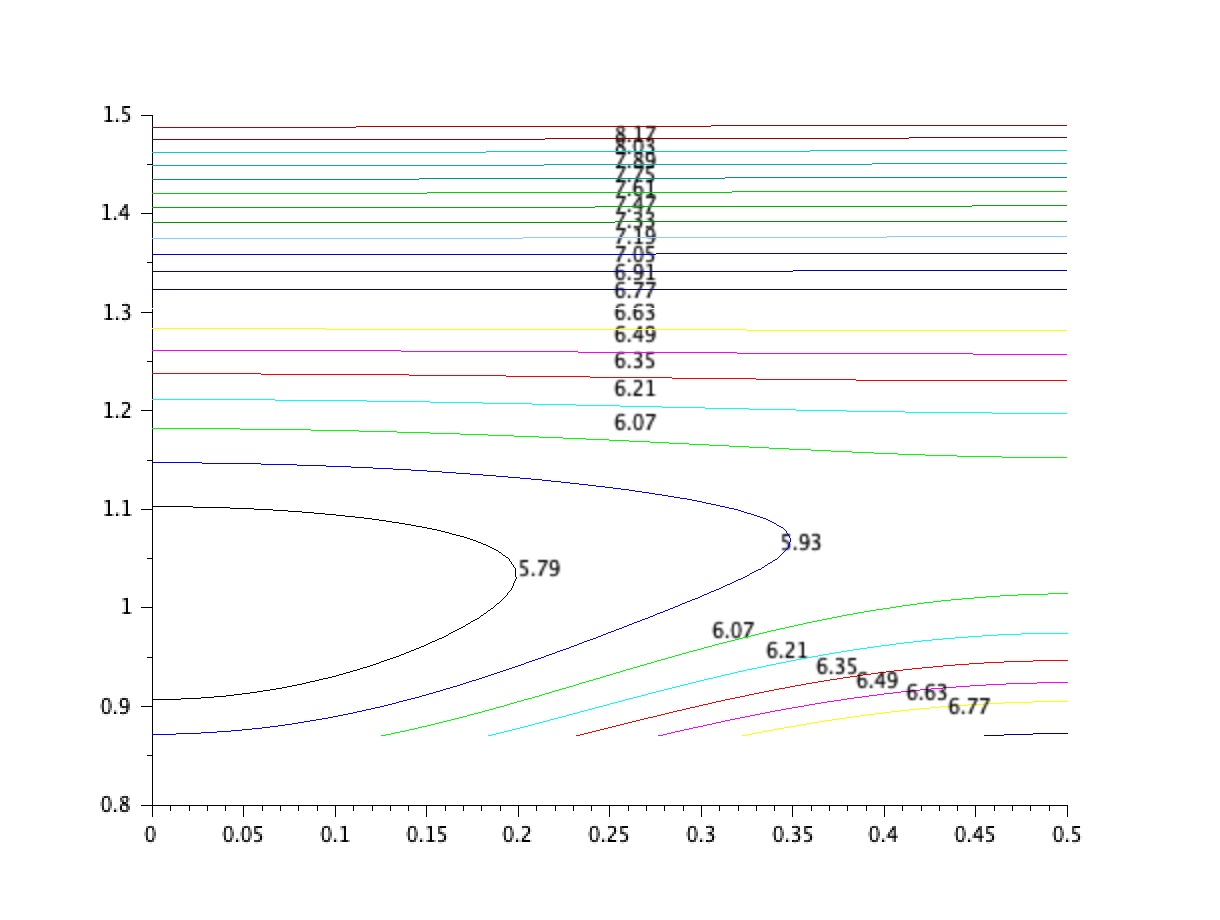}\quad  \includegraphics[width=8cm]{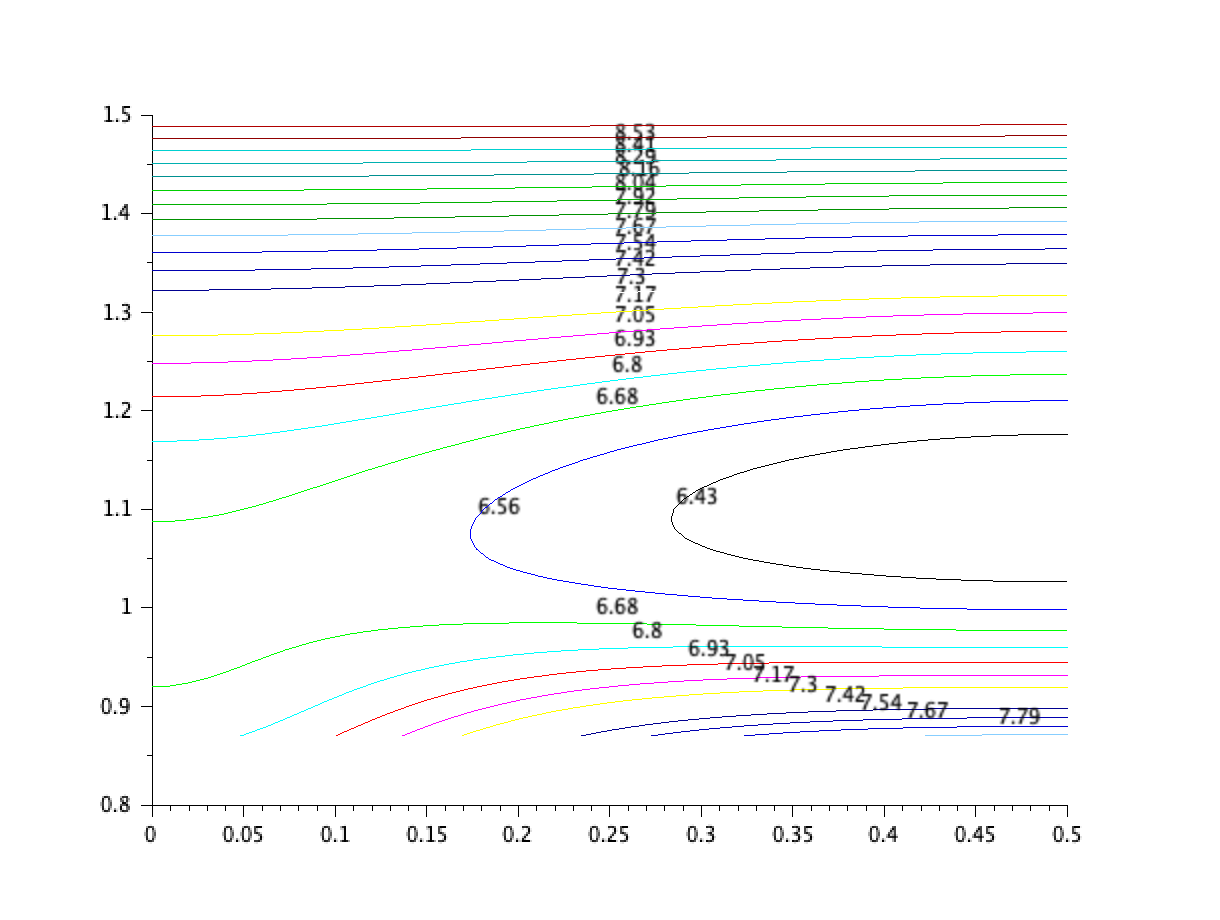} \quad \includegraphics[width=8cm]{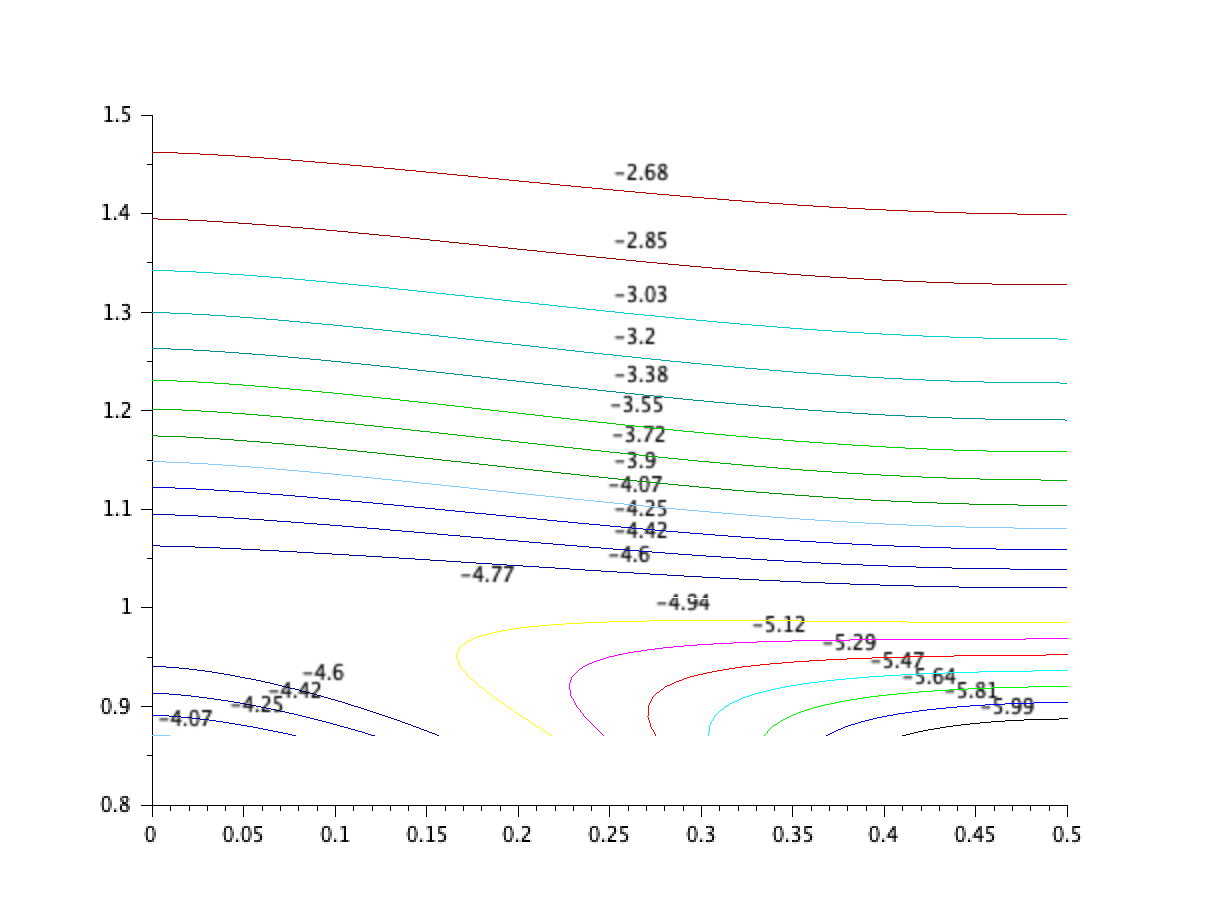}\quad \includegraphics[width=8cm]{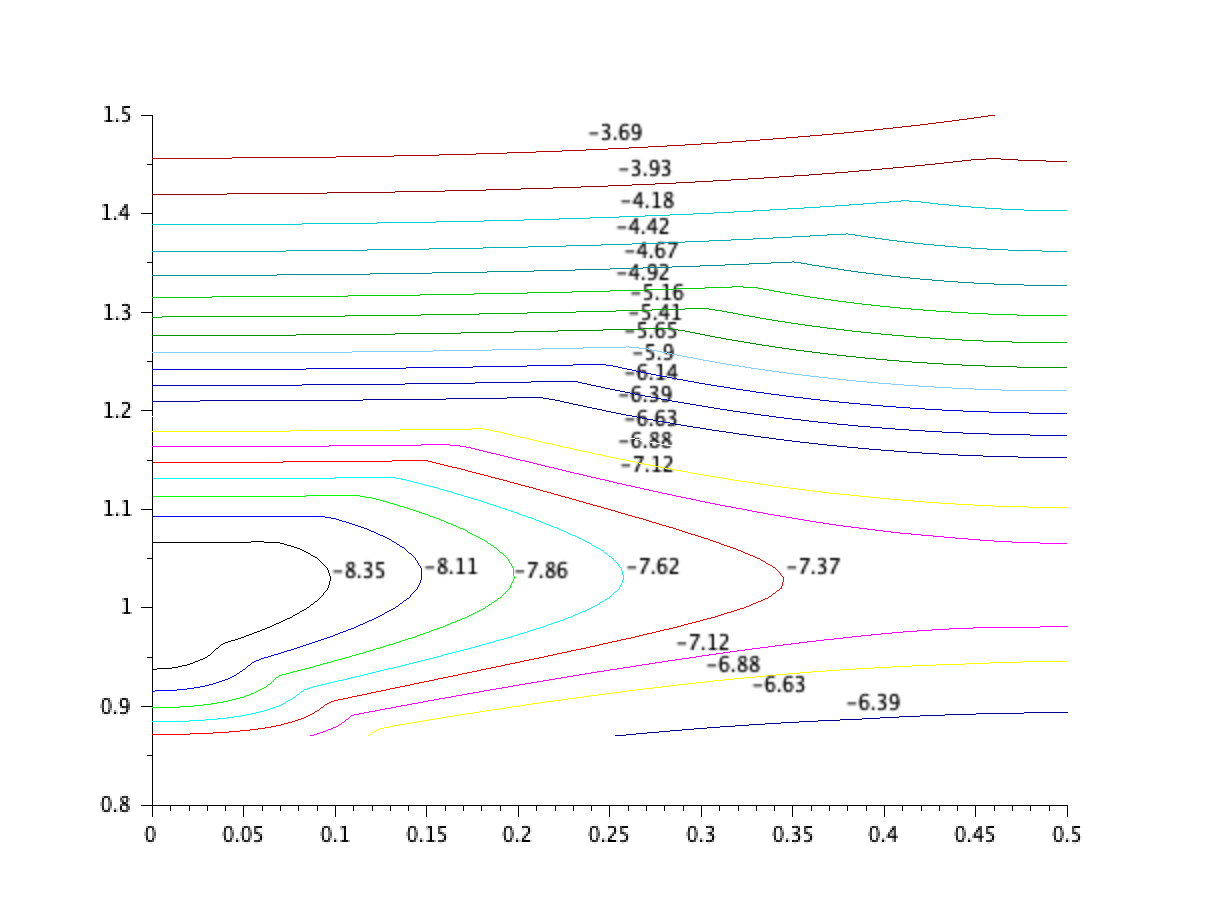}
\caption{Level curves of $L\mapsto \zeta_{L,\|\cdot\|_4}(6)$ (top left, with an optimal square lattice),  $L\mapsto \zeta_{L,\|\cdot\|_8}(6)$  (top right, with an optimal lattice parametrized by $\approx(1/2,1.098)$, $e_{\|\cdot\|_{1.5}}$ (bottom left, with an optimal triangular lattice) and $e_{\|\cdot\|_{\infty}}$ (bottom right, with an optimal square lattice) on a part of $\mathcal{D}$.}
\label{fig:contour}
\end{figure}

According to our simulations, we can write the following conjecture concerning the optimality of the triangular and the square lattice. We could of course write a more general conjecture from our numerical observations but since the shape of the other minimizers seems a bit vague, we prefer to only focus on $\mathsf{A}_2$ and $\Z^2$.

\begin{Conjecture}[\textbf{Minimality for Epstein zeta functions and Lennard-Jones energy}]\label{conj:lattice}
Let $s>2$. There exists $1<p_1<p_2$ and $1<p_3<p_4$ such that, up to rotation,
\begin{enumerate}
\item concerning $L\mapsto \zeta_{L,\|\cdot\|_p}(s)$ in $\mathcal{L}_2(1)$,
\begin{itemize}
\item[•] for all $p\in [1,2]$, the minimizer is $\sqrt{\frac{2}{\sqrt{3}}}\mathsf{A}_2$,
\item[•] for all $p\in (p_1,p_2)$, the minimizer is $\Z^2$.
\end{itemize}
\item concerning $L\mapsto \textnormal{E}_{\|\cdot\|_p}^{\textnormal{LJ}}(L)$ in $\mathcal{L}_2$,
\begin{itemize}
\item[•] for all $p\in (p_3,2]$, the minimizer is $\lambda_{\mathsf{A}_2} \mathsf{A}_2$,
\item[•] for all $p\in \{1\}\cup(p_4,\infty]$, the minimizer is $\lambda_{\Z^2} \Z^2$.
\end{itemize}
\end{enumerate}
\end{Conjecture}

\begin{remark}[\textbf{Open problems}]
Many other attractive-repulsive potentials can be taken into consideration: differences of Yukawa, Gaussians or Coulomb potentials as it was done very recently by Sun, Wei and Zou in \cite{SunWeiZou24} in the classical euclidean norm case. For all these potentials, two interesting questions emerge:
\begin{itemize}
\item what is the minimizer of their lattice energies in $\mathcal{L}_2$? One could guess that the minimizer follows again the above numerical observations. Without the homogeneity ``trick" explained above, this seems technically difficult to handle.
\item what are the minimizers at fixed arbitrary density of their lattice energies? Here, a phase diagram -- as the one described in \cite{Beterloc,SamajTravenecLJ,LuoWei22,SunWeiZou24} in the euclidean case -- could be done. This seems to be numerically accessible but technically challenging.
\item what is the expected Wulff shape in the limit of
a large number of particles for the Lennard-Jones
potential and an arbitrary norm? In the isotropic triangular lattice case, this question has been partially addressed in \cite{NinPetrache}. In the general $p$-norm case, one might speculate that the boundary of the Wulff shape would be hexagonal (resp. octagonal) for $\|\cdot\|_p$ when $p\in [p_3,2)$ (resp. for $p\in \{1\}\cup(p_4,\infty]$).
\end{itemize}
\end{remark}

\noindent \textbf{Acknowledgement:} The authors thank Universit\'e Claude Bernard Lyon 1 for the meeting opportunity given to us via the teaching unit ``Travaux d'initiative personnelle encadrés" during spring semester 2024. Parts of this paper were already included in Camille Furlanetto's work \cite{FurlanettoTIPE} (in French). We also want to thank the anonymous referee for their valuable comments.

{\small \bibliographystyle{plain}

\begin{thebibliography}{10}

\bibitem{BetTheta15}
L.~B{\'e}termin.
\newblock {Two-dimensional Theta Functions and Crystallization among Bravais
  Lattices}.
\newblock {\em SIAM J. Math. Anal.}, 48(5):3236--3269, 2016.

\bibitem{Beterloc}
L.~B{\'e}termin.
\newblock {Local variational study of 2d lattice energies and application to
  Lennard-Jones type interactions}.
\newblock {\em Nonlinearity}, 31(9):3973--4005, 2018.

\bibitem{BDefects20}
L.~B{\'e}termin.
\newblock Effect of periodic arrays of defects on lattice energy minimizers.
\newblock {\em Annales Henri Poincar{\'e}}, 22:2995--3023, 2021.

\bibitem{LBbonds21}
L.~B{\'e}termin.
\newblock On energy ground states among crystal lattice structures with
  prescribed bonds.
\newblock {\em J. Phys. A}, 54(24):245202, 2021.

\bibitem{LBComputerLJ23}
L.~B{\'e}termin.
\newblock {Optimality of the triangular lattice for Lennard--Jones type lattice
  energies: a computer-assisted method}.
\newblock {\em Journal of Physics A: Mathematical and General}, 56:145204,
  2023.

\bibitem{BFMaxTheta20}
L.~B{\'e}termin and M.~Faulhuber.
\newblock {Maximal Theta Functions - Universal Optimality of the Hexagonal
  Lattice for Madelung-Like Lattice Energies}.
\newblock {\em Journal d'Analyse Math{\'e}matique}, 149:307--341, 2023.

\bibitem{BDLPSquare}
L.~B{\'e}termin, L.~De Luca, and M.~Petrache.
\newblock {Crystallization to the square lattice for a two-body potential}.
\newblock {\em Arch. Ration. Mech. Anal.}, 240:987--1053, 2021.

\bibitem{OptinonCM}
L.~B{\'e}termin and M.~Petrache.
\newblock {Optimal and non-optimal lattices for non-completely monotone
  interaction potentials.}
\newblock {\em Anal. Math. Phys.}, 9(4):2033--2073, 2019.

\bibitem{Betermin:2014fy}
L.~B{\'e}termin and P.~Zhang.
\newblock {Minimization of energy per particle among Bravais lattices in
  $\R^2$: Lennard-Jones and Thomas-Fermi cases}.
\newblock {\em {Commun. Contemp. Math.}}, 17(6):1450049, 2015.

\bibitem{BlancLebris}
X.~Blanc and C.~Le Bris.
\newblock {Periodicity of the infinite-volume ground state of a one-dimensional
  quantum model}.
\newblock {\em Nonlinear Analysis T.M.A.}, 48(6):791--803, 2002.

\bibitem{BlancLewin-2015}
X.~Blanc and M.~Lewin.
\newblock {The Crystallization Conjecture: A Review}.
\newblock {\em EMS Surv. in Math. Sci.}, 2:255--306, 2015.

\bibitem{Brass}
P.~Brass.
\newblock Erd\"os distance problems in normed spaces.
\newblock {\em Computational Geometry}, 6:195--214, 1996.

\bibitem{Canizo24}
J.~A. Ca{\~n}izo and A.~Ramos-Lora.
\newblock Discrete minimizers of the interaction energy in collective behavior:
  a brief numerical and analytic review.
\newblock{\em In: Carrillo, J.A., Tadmor, E. (eds) Active Particles, Volume 4. Modeling and Simulation in Science, Engineering and Technology.} Birkhäuser, Cham, 2024.

\bibitem{Cassels}
J.W.S. Cassels.
\newblock {On a Problem of Rankin about the Epstein Zeta-Function}.
\newblock {\em Proceedings of the Glasgow Mathematical Association}, 4:73--80,
  7 1959.

\bibitem{CicaleseLeonardi20}
M.~Cicalese and G.P. Leonardi.
\newblock {Maximal Fluctuations on Periodic Lattices: An Approach via
  Quantitative Wulff Inequalities}.
\newblock {\em Commun. Math. Phys.}, 375:1931--1944, 2020.

\bibitem{Ciftja22}
O.~Ciftja.
\newblock {Results for an anisotropic Coulomb interaction potential}.
\newblock {\em Results in Physics}, 43:106052, 2022.

\bibitem{DPS17}
E. Davoli, P. Piovano, and U. Stefanelli.
\newblock Sharp $N^{3/4}$ Law for the Minimizers of the
Edge-Isoperimetric Problem on the Triangular Lattice.
\newblock{\em Journal of Nonlinear Science} 27(2):627--660, 2017.

\bibitem{NinLuca25}
G.~Del Nin and L.~De Luca.
\newblock{The square sticky disk: crystallization and Gamma-convergence to the octagonal anisotropic perimeter}.
\newblock{ Preprint, Arxiv:2503.20439}, 2025.

\bibitem{NinLucaSoft}
G.~Del Nin and L.~De Luca.
\newblock{A Crystallization Result in Two Dimensions for a Soft Disc Affine Potential}.
\newblock{ In: Franceschi, V., Pluda, A., Saracco, G. (eds) \textit{Anisotropic Isoperimetric Problems and Related Topics. INdAM 2022.} Springer INdAM Series, vol 62, 2022.}

\bibitem{NinPetrache}
G.~Del Nin and M.~Petrache.
\newblock{Continuum limits of discrete isoperimetric problems and Wulff shapes in lattices and quasicrystal tilings.}
\newblock{\em Calculus of Variations and Partial Differential Equations} 226:61, 2022



\bibitem{DelucaFriesecke-2018}
L.~{De Luca} and G.~Friesecke.
\newblock {Crystallization in {T}wo {D}imensions and a {D}iscrete
  {G}auss--{B}onnet {T}heorem}.
\newblock {\em J. Nonlinear Sci.}, 28(1):69--90, 2018.

\bibitem{DeLucaFrieseckeCassification}
L.~De Luca and G.~Friesecke.
\newblock {Classification of Particle Numbers with Unique Heitmann--Radin
  Minimizer}.
\newblock {\em J. Stat. Phys.}, 167:1586--1592, 2017.

\bibitem{Diananda}
P.~H. Diananda.
\newblock {Notes on Two Lemmas concerning the Epstein Zeta-Function}.
\newblock {\em Proceedings of the Glasgow Mathematical Association},
  6:202--204, 7 1964.

\bibitem{Engel}
P.~Engel.
\newblock {\em {Geometric Crystallography. An Axiomatic Introduction to
  Crystallography}}.
\newblock R. Reidel Publishing Compagny, 1942.

\bibitem{Eno2}
V.~Ennola.
\newblock {A Lemma about the Epstein Zeta-Function}.
\newblock {\em Proceedings of The Glasgow Mathematical Association},
  6:198--201, 1964.

\bibitem{AuyeungFrieseckeSchmidt-2012}
G.~Friesecke, B.~Schmidt, and Y.~A. Yeung.
\newblock {Minimizing atomic configurations of short range pair potentials in
  two dimensions: crystallization in the {W}ulff shape}.
\newblock {\em Calc. Var. Partial Differential Equations}, 44(1-2):81--100,
  2012.
  
  
  
\bibitem{FKS2020}
  M.~Friedrich, L.~Kreutz and B.~Schmidt.
\newblock{Emergence of Rigid Polycrystals from Atomistic Systems with Heitmann–Radin Sticky Disk Energy}.
\newblock {\em Archive for Rational Mechanics and Analysis}, 2021:627--698, 2021.
  
\bibitem{FurlanettoTIPE}
C.~Furlanetto.
\newblock R{\'e}seaux et cristallisation dans le plan pour un potentiel {\`a}
  sph{\`e}res dures.
\newblock {Rapport de TIPE}, Universit{\'e} Claude Bernard Lyon 1,
  \url{https://licence-math.univ-lyon1.fr/lib/exe/fetch.php?media=tipe:furlanetto_tipe_final.pdf},
  2024.

\bibitem{Rad1}
C.~S. Gardner and C.~Radin.
\newblock {The Infinite-Volume Ground State of the Lennard-Jones Potential}.
\newblock {\em Journal of Statistical Physics}, 20:719--724, 1979.

\bibitem{Grunbaum61}
B.~Gr\"unbaum.
\newblock {On a conjecture of H. Hadwiger}.
\newblock {\em Pacific J. Math.}, 11:215--219, 1961.

\bibitem{harborth}
H.~Harborth.
\newblock L\"osung zu problem 664a.
\newblock {\em Elem. Math.}, 29(14--15), 1974.

\bibitem{Rad2}
R.~C. Heitmann and C.~Radin.
\newblock {The Ground State for Sticky Disks}.
\newblock {\em J. Stat. Phys.}, 22:281--287, 1980.


\bibitem{Lamy15}
X.~Lamy.
\newblock Uniaxial symmetry in nematic liquid crystals.
\newblock {\em Ann. I. H. Poincar{\'e}}, 32:1125--1144, 2015.

\bibitem{LuoWei22}
S.~Luo and J.~Wei.
\newblock {On Minima of Difference of Epstein Zeta Functions and Exact
  Solutions to Lennard-Jones Lattice Energy}.
\newblock{\em J. Eur. Math. Soc.}, online first, 2025.

\bibitem{MPSS19}
E. Mainini, P. Piovano, B. Schmidt, and U. Stefanelli.
\newblock $N^{3/4}$ law in the cubic lattice.
\newblock{\em Journal of Statistical Physics} 176(6):1480--1499, 2019.

\bibitem{C0SM01205J}
{\'E}.~Marcotte, F.~H. Stillinger, and S.~Torquato.
\newblock Optimized monotonic convex pair potentials stabilize low-coordinated
  crystals.
\newblock {\em Soft Matter}, 7:2332--2335, 2011.

\bibitem{Mont}
H.~L. Montgomery.
\newblock {Minimal Theta Functions}.
\newblock {\em Glasg. Math. J.}, 30(1):75--85, 1988.

\bibitem{Rad3}
C.~Radin.
\newblock {The Ground State for Soft Disks}.
\newblock {\em J. Stat. Phys.}, 26(2):365--373, 1981.

\bibitem{RadinLowT}
C.~Radin.
\newblock Low temperature and the origin of crystalline symmetry.
\newblock {\em International Journal of Modern Physics B}, 1(5 and
  6):1157--1191, 1987.

\bibitem{Rankin}
R.~A. Rankin.
\newblock {A Minimum Problem for the Epstein Zeta-Function}.
\newblock {\em Proceedings of The Glasgow Mathematical Association},
  1:149--158, 1953.



\bibitem{Sandier_Serfaty}
E.~Sandier and S.~Serfaty.
\newblock {From the {G}inzburg-{L}andau {M}odel to {V}ortex {L}attice
  {P}roblems}.
\newblock {\em Comm. Math. Phys.}, 313(3):635--743, 2012.

\bibitem{Schmidt2013}
B. Schmidt.
\newblock{Ground states of the 2D sticky disc model: fine properties and $N^{\frac{3}{4}}$ law for the deviation from the asymptotic Wulff shape.}
\newblock{\em Journal of Statistical Physics} 153:727--738, 2013.


\bibitem{Sheargridcells}
T.~Stensola, H.~Stensola, M.B. Moser, and E.I. Moser.
\newblock Shearing-induced asymmetry in entorhinal grid cells.
\newblock {\em Nature}, 518:207--212, 2015.

\bibitem{SunWeiZou24}
J.~Sun, J.~Wei  and W.~Zou.
\newblock On lattice energy minimization problem for non-completely monotone functions and applications.
\newblock{\em Analysis and Mathematical Physics} 9:2033--2073, 2019.

\bibitem{Swanepoel2018}
K.~J. Swanepoel.
\newblock {\em Combinatorial Distance Geometry in Normed Spaces}, pages
  407--458.
\newblock Springer Berlin Heidelberg, Berlin, Heidelberg, 2018.

\bibitem{Crystal}
F.~Theil.
\newblock {A Proof of Crystallization in Two Dimensions}.
\newblock {\em Comm. Math. Phys.}, 262(1):209--236, 2006.

\bibitem{Torquato09}
S.~Torquato.
\newblock Inverse optimization techniques for targeted self-assembly.
\newblock {\em Soft Matter}, 5:1157, 2009.

\bibitem{SamajTravenecLJ}
I.~Trav\v{e}nec and L.~\v{S}amaj.
\newblock {Two-dimensional Wigner crystals of classical Lennard-Jones
  particles}.
\newblock {\em J. Phys. A: Math. Theor.}, 52(20):205002, 2019.

\bibitem{VN1}
W.J. Ventevogel.
\newblock {On the Configuration of Systems of Interacting Particle with Minimum
  Potential Energy per Particle}.
\newblock {\em Physica A-statistical Mechanics and Its Applications}, 92A:343,
  1978.

\bibitem{VN2}
W.J. Ventevogel and B.R.A. Nijboer.
\newblock {On the Configuration of Systems of Interacting Particle with Minimum
  Potential Energy per Particle}.
\newblock {\em Physica A-statistical Mechanics and Its Applications},
  98A:274--288, 1979.

\end{thebibliography}

}

\end{document}